\documentclass[11pt,onecolumn]{IEEEtran}

\usepackage[
            CJKbookmarks=true,
            bookmarksnumbered=true,
            bookmarksopen=true,
            colorlinks=true,
            citecolor=red,
            linkcolor=blue,
            anchorcolor=red,
            urlcolor=blue,
            pdfauthor={Kiryung Lee, Yanjun Li, Kyong Hwan Jin, and Jong Chul Ye},
            pdfstartview=FitH,
            ]{hyperref}

\usepackage[cmex10]{amsmath}
\usepackage{amsxtra,amssymb,amsthm,amsfonts}
\usepackage{mathtools}
\usepackage[matha,mathx]{mathabx}
\usepackage{mathrsfs}
\usepackage[usenames]{color}
\usepackage{umoline}
\usepackage{datetime}
\usepackage[margin=1in]{geometry}
\usepackage[ruled,boxed,linesnumbered]{algorithm2e}
\usepackage{setspace}
\usepackage{bbm}

\usepackage{enumerate}
\usepackage{graphicx}
\usepackage{subfig}
\usepackage{amssymb}
\usepackage{array}
\usepackage{multirow}

\usepackage{cleveref}
\crefname{equation}{}{}

\usepackage{tikz}
\usepackage{pgfplots}

\newtheorem{lemma}{Lemma}[section]

\newtheorem{theorem}[lemma]{Theorem}
\newtheorem{cor}[lemma]{Corollary}

\newtheorem{rem}[lemma]{Remark}

\newcommand{\re}{\begin{rem}\rm}
\newcommand{\mar}{\end{rem}}

\newtheorem{defi}[lemma]{Definition}

\newcommand{\fo}{\begin{eqnarray*}}
\newcommand{\mel}{\end{eqnarray*}}

\newcommand{\cz}{{\mathbb C}}

\newcommand{\calG}{{\mathcal G}}

\newcommand{\qd}{\end{proof}\vspace{0.5ex}}

\newcommand\tnorm[1]{\left\vert\xspace\left\vert\xspace\left\vert\mskip2mu
#1\mskip2mu \right\vert\xspace\right\vert\xspace\right\vert}

\makeatletter
\newcommand{\opnorm}{\@ifstar\@opnorms\@opnorm}
\newcommand{\@opnorms}[1]{%
  \left|\mkern-1.5mu\left|\mkern-1.5mu\left|
   #1
  \right|\mkern-1.5mu\right|\mkern-1.5mu\right|
}
\newcommand{\@opnorm}[2][]{%
  \mathopen{#1|\mkern-1.5mu#1|\mkern-1.5mu#1|}
  #2
  \mathclose{#1|\mkern-1.5mu#1|\mkern-1.5mu#1|}
}
\makeatother

\newcommand{\norm}[1]{\Vert#1\Vert}

\newcommand{\transpose}{\top}
\newcommand{\diag}{\mathrm{diag}}

\newcommand{\supp}[1]{\mathrm{supp}(#1)}
\newcommand{\gsupp}[1]{\mathrm{gsupp}(#1)}

\newcommand{\argmin}{\mathop{\rm argmin}}

\DeclareRobustCommand{\stirling}{\genfrac\{\}{0pt}{}}

\DeclareMathOperator*{\minimize}{\mathrm{minimize}}

\title{Unified Theory for Recovery of Sparse Signals \\ in a General Transform Domain}
\author{
Kiryung Lee, Yanjun Li~\IEEEmembership{Student Member,~˜IEEE}, Kyong Hwan Jin, \\ and Jong Chul Ye,~\IEEEmembership{Senior Member,~˜IEEE}
\thanks{
This work was supported in part by the National Science Foundation under grants IIS 14-47879
and Korea Science and Engineering Foundation under grants NRF-2013M3A9B2076548 and NRF-2016R1A2B3008104. K. Lee is with the School of Electrical and Computer Engineering, Georgia Institute of Technology, Atlanta, GA 30332 USA (e-mail: kiryung@ece.gatech.edu).
Y. Li is with the Coordinated Science Laboratory and the Department
of Electrical and Computer Engineering, University of Illinois at
Urbana-Champaign, Urbana, IL 61801, USA (e-mail: yli145@illinois.edu).
K.H. Jin is with the Biomedical Imaging Group, \'{E}cole Polytechnique
F\'{e}d\'{e}rale de Lausanne, 1015 Lausanne, Switzerland (e-mail: kyong.jin@epfl.ch).
J.C. Ye is with the Department of Bio and Brain Engineering, Korea Advanced Institute of Science and Technology (KAIST),
Daejon 305-701, Korea (e-mail: jong.ye@kaist.ac.kr).
}}

\begin{document}

\maketitle

\begin{abstract}
  Compressed sensing provided a data-acquisition paradigm for sparse signals.  Remarkably, it has been shown that practical algorithms provide robust recovery from noisy linear measurements acquired at a near optimal sampling rate.  In many real-world applications, a signal of interest is typically sparse not in the canonical basis but in a certain transform domain, such as wavelets or the finite difference.  The theory of compressed sensing was extended to the analysis sparsity model but known extensions are limited to specific choices of sensing matrix and sparsifying transform.  In this paper, we propose a unified theory for robust recovery of sparse signals in a general transform domain by convex programming.  In particular, our results apply to general acquisition and sparsity models and show how the number of measurements for recovery depends on properties of measurement and sparsifying transforms.  Moreover, we also provide extensions of our results to the scenarios where the atoms in the transform has varying incoherence parameters and the unknown signal exhibits a structured sparsity pattern.  In particular, for the partial Fourier recovery of sparse signals over a circulant transform, our main results suggest a uniformly random sampling.  Numerical results demonstrate that the variable density random sampling by our main results provides superior recovery performance over known sampling strategies.
\end{abstract}


\section{Introduction}
\label{sec:intro}

The theory of compressed sensing (CS) \cite{donoho2006compressed,candes2006robust} provided a new data-acquisition paradigm for sparse signals.
Remarkably, it has been shown that practical algorithms are guaranteed to reconstruct the unknown sparse signal
from the linear measurements taken at a provably near optimal rate.
Reconstruction algorithms with performance guarantees include
modern optimization algorithms for $\ell_1$-norm-based convex optimization formulations (e.g., \cite{beck2009fast,boyd2011distributed}) and iterative greedy algorithms (e.g., \cite{needell2009cosamp,dai2009subspace,blumensath2009iterative,foucart2011hard}).

The canonical sparsity model in CS assumes that the unknown signal $f \in \cz^d$ is $s$-sparse in the standard coordinate basis.
In other words, $\norm{f}_0 \leq s$, where $\norm{\cdot}_0$ counts the number of nonzero elements.
The acquisition process in CS is linear and represented by a sensing matrix $A \in \cz^{m \times d}$
so that the $m$ linear measurements in $b \in \cz^m$ is given by
\[
b = A f + w,
\]
where $w \in \cz^m$ denotes additive noise to the measurements and satisfies $\norm{w}_2 \leq \epsilon$.
For certain random sensing matrices, it was shown that an estimate $\hat{f}$ given by
\begin{equation}
\label{eq:bpdn}
\hat{f} = \argmin_{\tilde{f} \in \cz^d} \norm{\tilde{f}}_1 \quad \mathrm{subject~to} \quad \norm{b - A \tilde{f}}_2 \leq \epsilon
\end{equation}
satisfies $\norm{\hat{f} - f}_2 \leq c_1 \epsilon$ with high probability,
provided that $m \geq C s \log^\alpha d$ for some $\alpha \in \mathbb{N}$ and numerical constants $C$ and $c_1$. In particular, in the noiseless case ($\epsilon = 0$), the estimate $\hat{f}$ coincides with the ground truth signal $f$.
For example, Candes and Tao \cite{candes2005decoding}
showed that the above guarantees hold for a Gaussian sensing matrix $A$ whose entries are i.i.d. following $\mathcal{N}(0,1/m)$
via the restricted isometry property (RIP).
Recent results with a sharper sample complexity of $m \geq 2 s \log(d/s)$
were derived using the Gaussian width of a tangent cone \cite{chandrasekaran2012convex,amelunxen2014living}.

In fact, the original idea of compressed sensing \cite{bresler1999image} was motivated by a need to accelerate various imaging modalities.
The sensing matrix $A$ in these applications takes observations in a measurement transform domain, i.e.
\begin{equation}
\label{eq:structred_sensing_mtx}
A = \sqrt{\frac{n}{m}} S_\Omega \Psi,
\end{equation}
where $\Psi \in \cz^{n \times d}$ is the matrix representation of the measurement transform and the sampling operator $S_\Omega: \cz^n \to \cz^m$ takes the $m$ elements indexed by $\Omega = \{\omega_1,\omega_2,\ldots,\omega_m\}$.
For example, when $\Psi$ is a discrete Fourier transform (DFT) matrix, $A$ is a partial Fourier matrix.
The aforementioned near optimal performance guarantees were shown
for a partial Fourier sensing matrix $A$ obtained using a random set $\Omega$ \cite{candes2006robust,candes2006near,rudelson2008sparse,candes2011probabilistic}
and generalized for the case where the rows of $\Psi \in \cz^{n \times d}$ correspond to an incoherent tight frame for $\cz^d$ \cite{rauhut2010compressive}.

However, in numerous imaging applications, a signal of interest is not sparse in the standard coordinate basis.
The theory of compressed sensing was accordingly extended to the so-called synthesis and analysis sparsity models \cite{candes2011compressed,nam2013cosparse,giryes2014greedy}.
The synthesis sparsity model assumes that $f \in \cz^d$ is represented as a linear combination of few atoms in a dictionary $D \in \cz^{d \times N}$.
Equivalently, $f$ is represented as $f = D u$ with an $s$-sparse coefficient vector $u \in \cz^N$.
Compressed sensing with the synthesis sparsity model can be interpreted as
conventional compressed sensing of $u$ using a sensing matrix $A D$ where $u$ is $s$-sparse in the standard basis.
In particular when $A$ is a Gaussian matrix and $D$ is an orthogonal matrix,
conventional performance guarantees carry over to the synthesis sparsity model.
On the other hand, the analysis sparsity model, which is motivated from sparse representation in harmonic analysis \cite{mallat2008wavelet},
assumes that the transform $\Phi f \in \cz^N$ of $f$ via $\Phi \in \cz^{N \times d}$ is $s$-sparse.
In fact, a signal of interest in practical applications often follows the analysis sparsity model
via various transforms including finite difference and wavelet \cite{mallat2008wavelet}, contourlet \cite{do2005contourlet},
curvelet \cite{starck2002curvelet}, and Gabor transforms.
Thus, the analysis sparsity model has been used as an effective regularizer for
classical inverse problems in signal processing (e.g., denoising and deconvolution)
and compressed sensing imaging (e.g., \cite{lustig2008compressed}).


Unlike the previous results with the canonical sparsity model,
the theory of compressed sensing with the analysis model has been relatively less explored and known results are limited to specific cases.
Candes et al. \cite{candes2011compressed} considered the recovery of $f \in \cz^d$ such that $\Phi f$ is $s$-sparse for a transform $\Phi \in \cz^{N \times d}$ satisfying $\Phi^* \Phi = I_d$, i.e. the columns of $\Phi^*$ correspond to a tight frame for $\cz^d$. They showed that
\begin{equation}
\label{eq:bpdn_tfm}
\hat{f} = \argmin_{\tilde{f} \in \cz^d} \norm{\Phi \tilde{f}}_1 \quad \mathrm{subject~to} \quad \norm{A \tilde{f} - b}_2 \leq \epsilon
\end{equation}
has the error bound given by $\norm{\hat{f} - f}_2 \leq c_1 \epsilon$,
provided that the sensing matrix $A \Phi^*$ satisfies the RIP.
In particular for a Gaussian sensing matrix $A$, their performance guarantee holds with high probability for $m = O(s \log(d/s))$ for any $\Phi$ satisfying $\Phi^* \Phi = I_d$.
Indeed the performance guarantee by Candes et al. \cite{candes2011compressed} applies beyond the case of a Gaussian sensing matrix $A$.
Krahmer and Ward \cite{krahmer2011new} showed that if $A \in \cz^{m \times d}$ and $D \in \cz^{d \times N}$ satisfy the RIP and $\varepsilon \in \mathbb{R}^d$ is a Rademacher sequence with random $\pm 1$ entries, then $A \diag(\varepsilon) D$ satisfies the RIP, where $\diag(\varepsilon) \in \cz^{d \times d}$ denotes the diagonal matrix whose diagonal entries are $\varepsilon$.
However, applying a random sign before the acquisition might not be feasible in certain applications.
In another line of research, it was shown \cite{nam2013cosparse,giryes2014greedy} that greedy algorithms for the analysis sparsity model
provide performance guarantees if the sensing matrix $A$ is a near isometry when acting on all transform-sparse $f$ such that $\Phi f$ is sparse.
Again for $\Phi$ satisfying $\Phi^* \Phi = I_d$, this condition on $A$ is less demanding than the RIP of $A \Phi^*$ since the latter implies the former. However, it has not been studied how such a relaxation translates into less-demanding sample complexity.

A special analysis sparsity model associated with the finite difference transform $\Phi$ has been of particular interest in signal processing and imaging applications. The corresponding convex surrogate $\norm{\Phi f}_1$, known as the total variation (TV), has been popularly used as an effective regularizer for solving inverse problems.
Needell and Ward \cite{needell2013stable} provided performance guarantees for TV minimization with a partial Fourier sensing matrix $A$
in terms of the RIP of $A D$ with a Haar wavelet dictionary $D$.
Their result was further refined by Krahmer and Ward \cite{krahmer2014stable}
with a clever idea of variable density sampling adopting the local incoherence parameters.
Remarkably, Krahmer and Ward \cite{krahmer2014stable} provided performance guarantees at the sample complexity of $m = O(s \log^3 s \log^2 d)$.
However, these results on TV minimization rely on a special embedding theorem
that relates the Haar transform to the finite difference transform, which holds only for signals of dimension two or higher.
Therefore, the results by Needell and Ward \cite{needell2013stable} and by Krahmer and Ward \cite{krahmer2014stable}
do not apply to the 1D case and more importantly do not generalize to other sparsifying transforms.

Recently, inspired by analogous results for the canonical sparsity model \cite{chandrasekaran2012convex,amelunxen2014living},
Kabanava et al. \cite{kabanava2015robust} derived a performance guarantee of TV minimization with a Gaussian sensing matrix $A$
at the sample complexity of $m > d \,[1 - \{ 1-(s+1)/d \}^2/\pi]$.
Cai and Xu \cite{cai2015guarantees} showed a similar result with $m \geq C \sqrt{s d} \log d$.
Compared to the previous result by Krahmer and Ward \cite{krahmer2014stable},
the above results \cite{kabanava2015robust,cai2015guarantees} apply to the 1D case but at a significantly suboptimal sample complexity.
More importantly, the use of a Gaussian sensing matrix might not be relevant to practical applications.
Along a similar analysis strategy, Kabanava and Rauhut \cite{kabanava2015analysis} showed performance guarantees for \eqref{eq:bpdn_tfm}
with a Gaussian sensing matrix $A \in \mathbb{R}^{m \times d}$ and a frame analysis operator $\Phi \in \mathbb{R}^{N \times d}$ roughly at the sample complexity of $m \geq 2 \kappa s \log (2N/s)$,
where $\kappa$ denotes the ratio of upper and lower frame bounds, i.e. $\kappa$ is the condition number of the frame operator $\Phi^* \Phi$.
The sample complexity of this result for a tight frame is near optimal.
However, their analysis is restricted to a Gaussian sensing matrix $A$ and does not generalize to other sensing matrices.

\subsection{Contributions}

Our main contribution is to derive performance guarantees for compressed sensing of analysis-sparse vectors by \eqref{eq:bpdn_tfm} for general classes of sensing matrix $A$ given in the form of \eqref{eq:structred_sensing_mtx} with measurement transform $\Psi$ and (redundant) analysis transform $\Phi$.
Unlike the previous works, performance guarantees in this paper apply without being restricted to a particular choice of $\Phi$ and $\Psi$\footnote{A Gaussian sensing matrix is also obtained in the form of \eqref{eq:structred_sensing_mtx} by choosing $\Psi$ as a Gaussian matrix.}.
The number of measurements implying these guarantees depends on certain properties of $(\Phi,\Psi)$ and this result identifies a class of measurement and sparsity models allowing recovery at a near optimal sampling rate.
Moreover, when $\Psi$ and $\Phi$ have a few strongly correlated atoms,
adopting the idea by Krahmer and Ward \cite{krahmer2014stable},
we propose to acquire linear measurements using random sampling with respect to
a variable density designed with the correlations between $\Psi$ and $\Phi$.
This modified acquisition enables recovery at a lower sampling rate.
We also extend the results to group sparsity models. This extension applies to various popular regularized recovery methods including the isotropic-total-variation minimization.
In special cases where $\Psi$ is Fourier and $\Phi$ is a circulant matrix, our theory suggests a uniformly random sampling or its variation. For example, for the total variation minimization, unlike the common belief in practice, a sampling strategy that combines the acquisition of the lowest frequency and uniformly random sampling on the other frequencies provides better reconstruction than known variable density random sampling strategies.

Our main idea is inspired by a previous work on structured matrix completion by Chen and Chi \cite{chen2014robust}.
They showed that a structured low-rank matrix (e.g., a low-rank Hankel matrix)
is successfully recovered from partially observed entries by minimizing the nuclear norm.
Indeed, their structured low-rank matrix completion can be interpreted as follows.
The unknown structured matrix $M$ is given as the image $T x$ of the generator $x$ via a linear map $T$.
Then the completion of the structured matrix $M$ with the low-rankness prior
is equivalent to the completion of the generator $x$ with low-rankness in the transform domain via $T$.
Similarly, compressed sensing with the analysis sparsity model is equivalent to
the recovery of $x \in \cz^n$ with the prior that $T x \in \cz^N$ is $s$-sparse
where the transform $T$ is given as $T = \Phi \Psi^\dagger$.
Therefore the two problems are analogous to each other in the sense that
their priors correspond to atomic sparsity \cite{chandrasekaran2012convex} in their respective transform domains.
Moreover, in both the structured low-rank matrix model and the transform-domain sparsity model, $T$ is not necessarily surjective, which implies that $T T^*$ may be rank-deficient.
This violates an important technical condition known as the \textit{isotropy} property
and many crucial steps in the proofs of existing performance guarantees for CS break down.
To overcome this difficulty, we adopt the clever idea by Chen and Chi \cite{chen2014robust} through the aforementioned analogy
and derive near optimal performance guarantees for the recovery of sparse signals in a transform domain without resorting to the isotropy property.
However, besides this similarity, our results are significantly different from the analogous results \cite{chen2014robust} in the following sense.
Chen and Chi \cite{chen2014robust} assumed that $T$ is restricted to a set of special linear operators that generate structured matrices and their analysis indeed relies critically on strong properties satisfied by such linear operators (e.g., $T$ needs to satisfy $T^* T = I_n$).
Contrarily, we only assume a mild condition that $T$ is injective and the resulting performance guarantees apply to more general cases.
For example, in compressed sensing with the analysis sparsity model,
$T = \Phi \Psi^\dagger$ is non-unitary if the measurement transform $\Psi$ is non-unitary (e.g., the Radon transform) or
the sparsifying transform $\Phi$ is non-unitary (e.g., biorthogonal wavelet and data-adaptive transforms).

We illustrate our theory through extensive Monte Carlo simulations. The variable density sampling in this paper provides an improved recovery performance over the previously suggested sampling strategies. For example, when $\Phi$ and $\Psi$ have atoms showing high correlation (e.g. Fourier and wavelet), our theory suggests to sample more densely in the lower frequencies, which enables successful recovery from fewer observations. Rather surprisingly, when $\Phi$ and $\Psi$ respectively correspond to the finite difference and the Fourier transforms, our theory suggests a special sampling density that always takes the lowest frequency and chooses the other frequency components randomly using the uniform density. This is contrary to the common belief in compressed sensing but the proposed sampling strategy turns out be more successful than previous works not only in theory but also empirically. These numerical results strongly support out new theory.

\subsection{Organization}

The rest of this paper is organized as follows.
The main results are presented in Sections~\ref{sec:main_result} and \ref{sec:main_result_vd}, where the proofs are deferred to Sections~\ref{sec:pf_main_result_noiseless} and~\ref{sec:pf_main_result_noisy}. We extend the theory to the group sparsity models in Section~\ref{sec:main_result_gs} and study a special case of circulant transforms in Section~\ref{sec:circulant}.
After demonstrating empirical observations in Section~\ref{sec:numres}, which supports our main results, we conclude the paper with some final remarks in Section~\ref{sec:concl}.

\subsection{Notations}

For a positive integer $N$, we will use a shorthand notation $[N]$ for the set $\{1,\ldots,N\}$.
For a vector $z \in \cz^N$, let $z[k]$ denote the $k$th element of $z$ for $k \in [N]$,
i.e. $z = [z[1], z[2], \ldots, z[N]]^\transpose$.
The Hadamard product of two vectors $x,y \in \cz^n$ is denoted by $x \odot y$. The circular convolution of two vectors $x,y \in \cz^n$ is denoted by $x \circledast y$. The Kronecker product of two matrices $A$ and $B$ is denoted by $A\otimes B$.
The operator norm from $\ell_p^n$ to $\ell_q^N$ will be denoted by $\norm{\cdot}_{p \to q}$.
For brevity, the spectral norm $\norm{\cdot}_{2 \to 2}$ will be written without subscript as $\norm{\cdot}$.
For a matrix $A$, its Hermitian transpose and its Moore-Penrose pseudo inverse are respectively written as $A^*$ and $A^\dagger$.

For $J \subset [N]$, the coordinate projection with respect to $J$, denoted by $\Pi_J$, is defined as
\[
(\Pi_J z)[k] :=
\begin{dcases}
z[k] & k \in J, \\
0 & \mathrm{otherwise}.
\end{dcases}
\]
The complex signum function, denoted by $\mbox{sgn}(\cdot): \cz^N \to \cz^N$, is defined as
\begin{equation}
\label{eq:sgn}
(\mbox{sgn}(z))[k] :=
\begin{dcases}
\frac{z[k]}{|z[k]|} & z[k] \neq 0, \\
0 & \mathrm{otherwise}.
\end{dcases}
\end{equation}
Let $e_1,\ldots,e_n$ denote the standard basis vectors for $\cz^n$.
In other words, $e_k$ is the $k$th column of the $n$-by-$n$ identity matrix.

\section{Recovery of Sparse Signals in a General Transform Domain}
\label{sec:main_result}

Let $T: \cz^n \to \cz^N$ be a linear map that we call a ``transform'' in this paper.
Let $\Omega = \{\omega_1,\omega_2,\ldots,\omega_m\}$ denote the multi-set of $m$ sampling indices out of $[n] := \{1,\ldots,n\}$
with possible repetition of elements.
Given $\Omega$, the sampling operator $S_\Omega: \cz^n \to \cz^m$ is defined so that
the $j$th element of $S_\Omega x \in \cz^m$ is the $\omega_j$th element of $x \in \cz^n$ for $j = 1,\ldots,m$.
We are interested in recovering an unknown signal $x \in \cz^n$ from its partial entries at $\Omega$
when the transform $T x$ is known $s$-sparse a priori, i.e. $\norm{T x}_0 \leq s$.
Compressed sensing with the analysis sparsity model is an instance of this problem formulation as shown in the following.
Recall that $\Psi$ is of full column rank.
Let $T: \cz^n \to \cz^N$ be defined by $T = \Phi \Psi^\dagger$ where $\Psi^\dagger$ denotes the Moore-Penrose pseudo inverse of $\Psi$.
Let $x \in \cz^n$ denote the vector containing fully sampled measurements, i.e. $x = \Psi f$.
Since $f = \Psi^\dagger x$, we have $\Phi f = \Phi \Psi^\dagger x = T x$.
Thus the recovery of $f$ from $b = \sqrt{n/m} S_\Omega \Psi f$ with the prior that $\Phi f$ is $s$-sparse
is equivalent to the recovery of $x$ from $b = \sqrt{n/m} S_\Omega x$ with the prior that $T x$ is $s$-sparse.

In the noise-free scenario where partial entries of $x$ are observed exactly,
we propose to estimate $x$ as the minimizer to the following optimization problem:
\begin{equation}
\label{eq:ell1min}
\minimize_{g \in \cz^n} \, \norm{T g}_1 \quad \mathrm{subject~to} \quad S_\Omega g = S_\Omega x.
\end{equation}
We provide a sufficient condition for recovering $x$ exactly by \eqref{eq:ell1min} in the following theorem.

\begin{theorem}
\label{thm:uniqueness}
Suppose $T: \cz^n \to \cz^N$ is injective.
Let $\gamma$ be defined by
\begin{equation}
\label{eq:defgamma}
\gamma := \argmin_{\tilde{\gamma} > 0} \norm{\tilde{\gamma} T^* T - I_n},
\end{equation}
where $\widetilde{T} = (T^\dagger)^*$.
Let $\mu$ be given by
\begin{equation}
\label{eq:incoherence}
\mu = \max_{k \in [n]} \max\left\{ n \norm{\gamma^{1/2} T e_k}_\infty^2, n \norm{\gamma^{-1/2} \widetilde{T} e_k}_\infty^2 \right\}.
\end{equation}

Let $\Omega = \{\omega_1,\ldots,\omega_m\}$ be a multi-set of random indices
where $\omega_k$s are i.i.d. following the uniform distribution on $[n]$.
Suppose $T x$ is $s$-sparse.
Then with probability $1 - e^{-\beta} - 3/n$, $x$ is the unique minimizer to \eqref{eq:ell1min} provided
\[
m \geq \frac{C(1+\beta) \mu s}{1 - \norm{\gamma T^* T - I_n}} \left[ \log N
+ \log \left( \norm{T}_{1 \to 2} \norm{T^\dagger}_{2 \to \infty} \right) \right].
\]
\end{theorem}

\begin{proof}[Proof of Theorem~\ref{thm:uniqueness}]
See Section~\ref{sec:pf_main_result_noiseless}.
\end{proof}

\begin{rem}
\rm
All the results in this section including Theorem~\ref{thm:uniqueness} as well as other similar results \cite{candes2011probabilistic,gross2011recovering,chen2014robust}, all derived using the golfing scheme, provide an \emph{instance} recovery guarantee that applies to a single arbitrary instance of $x$, which is a weaker result than the \emph{uniform} recovery guarantee that applies to the set of all transform-sparse signals. In compressed sensing with the canonical sparsity model, an instance guarantee is given from a fewer measurements than the uniform guarantee \cite{candes2011probabilistic} by a poly-log factor. For matrix completion problems, the RIP do not hold and known results \cite{gross2011recovering,chen2014robust} only provide an instance guarantee. We suspect that this is the case with the recovery problem in this paper.
\end{rem}

Note that $\norm{T}_{1 \to 2}$ denotes the largest $\ell_2$-norm among all columns of $T$.
Similarly, $\norm{T^\dagger}_{2 \to \infty} = \norm{\widetilde{T}}_{1 \to 2}$
denotes the largest $\ell_2$-norm among all columns of $\widetilde{T}$.
In a special case when $T^* T = I_n$, we have $T^\dagger = T^*$. Thus $\norm{T}_{1 \to 2} = \norm{T^\dagger}_{2 \to \infty} = 1$.
In general, if $\norm{T}_{1 \to 2} \norm{T^\dagger}_{2 \to \infty} = O(N^{\alpha_1})$ for some $\alpha_1 \in \mathbb{N}$
and
\[
\frac{1}{1 - \norm{\gamma T^* T  - I_n}} = O(\log^{\alpha_2} N)
\]
for some $\alpha_2 \in \mathbb{N}$,
then we get a performance guarantee at a near optimal scaling of sample complexity of
$m = O(\mu s \log^\alpha N)$, where $\alpha = \alpha_1 + \alpha_2$.
These are mild conditions and easily satisfied by transforms that arise in practical applications.

In the noisy scenario where partial entries of $x$ are observed with additive noise,
we propose to estimate $x$ by solving the following optimization problem:
\begin{equation}
\label{eq:ell1min_noisy}
\minimize_{g \in \cz^n} \, \norm{T g}_1
\quad \mathrm{subject~to} \quad \norm{S_{\Omega'} g - S_{\Omega'} x^\sharp}_2 \leq \epsilon,
\end{equation}
where $x^\sharp$ denotes a noisy version of $x$,
$\Omega'$ denotes the set of all unique elements in $\Omega$,
and $S_{\Omega'}^*$ is the adjoint of the sampling operator $S_{\Omega'}$ that fills missing entries at outside $\Omega'$ with 0.

\begin{theorem}
\label{thm:stability1}
Suppose the hypotheses of Theorem~\ref{thm:uniqueness} hold.
Let $\hat{x}$ be the minimizer to \eqref{eq:ell1min_noisy} with $x^\sharp$ satisfying
\begin{equation}
\label{eq:xharp_fid}
\norm{S_{\Omega'} (x - x^\sharp)}_2 \leq \epsilon.
\end{equation}
Then
\[
\norm{\hat{x} - x}_2 \leq
\frac{\sigma_{\max}(T)}{\sigma_{\min}(T)}
\cdot
\left\{ 2 + 28 \sqrt{N} \left( 3 n \norm{T}_{1 \to 2} \norm{T^\dagger}_{2 \to \infty} + 1 \right) \right\} \epsilon,
\]
\end{theorem}

\begin{proof}[Proof of Theorem~\ref{thm:stability1}]
See Section~\ref{subsec:pf:thm:stability1}.
\end{proof}

We can tighten the upper bound on the estimation error in Theorem~\ref{thm:stability1} when $T^* T$ is well conditioned.
To this end, we will use the following theorem, which is obtained by modifying \cite[Theorem~3.1]{lee2013oblique}.
\begin{theorem}[{A modification of \cite[Theorem~3.1]{lee2013oblique}}]
\label{thm:rboplike}
Suppose $T: \cz^n \to \cz^N$ and $\widetilde{T} = (T^\dagger)^*$ satisfy \eqref{eq:incoherence} with parameters $\mu$ and $\gamma$.
Let $\Omega = \{\omega_1,\ldots,\omega_m\}$ be a multi-set of random indices
where $\omega_k$s are i.i.d. following the uniform distribution on $[n]$.
Then with probability $1-\xi$, we have
\begin{equation}
\label{eq:rboplike}
\max_{|\widetilde{J}| \leq s} \left\| \Pi_{\widetilde{J}} \left(\frac{n}{m} T S_\Omega^* S_\Omega T^\dagger - T T^\dagger \right) \Pi_{\widetilde{J}} \right\| \leq \delta,
\end{equation}
provided
\begin{align}
\label{eq:rboplike:cond1}
m {} & \geq \frac{C_1 \delta^{-2} \mu s \log^2 s \log N \log m}{1 - \norm{\gamma T^* T  - I_n}},
\intertext{and}
\label{eq:rboplike:cond2}
m {} & \geq C_2 \delta^{-2} \mu s \log (\xi^{-1}).
\end{align}
\end{theorem}

\begin{proof}[Proof of Theorem~\ref{thm:rboplike}]
See Appendix~\ref{sec:pf:thm:rboplike}.
\end{proof}

Using Theorem~\ref{thm:rboplike},
we provide another sufficient condition for stable recovery of sparse signals in a transform domain,
which has a smaller noise amplification factor.

\begin{theorem}
\label{thm:stability2}
Suppose $T: \cz^n \to \cz^N$ is injective.
Suppose $T$ and $\widetilde{T} = (T^\dagger)^*$ satisfy \eqref{eq:incoherence} with parameters $\mu$ and $\gamma$.
Let $\Omega = \{\omega_1,\ldots,\omega_m\}$ be a multi-set of random indices
where $\omega_k$s are i.i.d. following the uniform distribution on $[n]$.
Let $\hat{x}$ be the minimizer to \eqref{eq:ell1min_noisy} with $x^\sharp$ satisfying \eqref{eq:xharp_fid}.
Then there exist numerical constants $C, c > 0$ for which the following holds.
With probability $1 - N^{-4}$, we have
\[
\norm{\hat{x} - x}_2 \leq
\frac{\sigma_{\max}(T)}{\sigma_{\min}(T)}
\cdot
\left[ 14\sqrt{N} + \frac{n}{m} \left\{ \norm{T}_{1 \to 2} \norm{T^\dagger}_{2 \to \infty} (|\Omega|-|\Omega'|) +1 \right\} \right] \epsilon,
\]
provided
\begin{equation}
\label{eq:noisy_samp_comp}
m \geq \frac{C \mu s \log^4 N}{1 - \norm{\gamma T^* T - I_n}}.
\end{equation}
Furthermore, if $T^* T = I_n$, then with probability $1 - N^{-4}$,
\[
\norm{\hat{x} - x}_2 \leq
\frac{\sigma_{\max}(T)}{\sigma_{\min}(T)}
\cdot
\left( 14\sqrt{N} + \frac{n}{m} R \right) \epsilon,
\]
where $R$ is the count of the most repeated elements in $\Omega$.
\end{theorem}

\begin{proof}[Proof of Theorem~\ref{thm:stability2}]
See Section~\ref{subsec:pf:thm:stability2}.
\end{proof}

\begin{rem}
\rm
By the definition of $\Omega'$, we have $|\Omega|-|\Omega'| \leq m-1$.
Suppose that $\norm{T}_{1 \to 2} \norm{T^\dagger}_{2 \to \infty} = O(1)$.
The resulting noise amplification factor by Theorem~\ref{thm:stability2} is $O(\sqrt{N} + n)$,
which is already smaller than $O(\sqrt{N} n)$ by Theorem~\ref{thm:stability1}.
In fact, the distribution of $|\Omega'|$ is explicitly given by
\[
\mathbb{P}\left( |\Omega'| = \ell \right) = \frac{\stirling{m}{\ell} n!}{n^m (n-\ell)!},
\]
where $\stirling{m}{\ell}$ denotes the Stirling number of the second kind defined by
\[
\stirling{m}{\ell} := \frac{1}{\ell!} \sum_{j=0}^\ell(-1)^{\ell-j} {\ell \choose j} j^m.
\]
It would be possible to compute a more tight probabilistic upper bound on $|\Omega|-|\Omega'|$ with its distribution.
However, because of the other factor $\sqrt{N}$, regardless of $|\Omega|-|\Omega'|$,
the noise amplification factor by Theorem~\ref{thm:stability2}
cannot be improved to $O(1)$ as shown for compressed sensing
with the canonical sparsity model \cite{candes2011probabilistic} or with a special analysis model \cite{krahmer2014stable}.
We admit that this suboptimality in noise amplification is a limitation of our analysis.
It will be interesting to see whether one can obtain near optimal noise amplification for a general transform $T$.
\end{rem}

\section{Incoherence-Dependent Variable Density Sampling}
\label{sec:main_result_vd}

In the results of the previous section, the number of measurements is proportional to the incoherence parameter $\mu$ in \eqref{eq:incoherence},
which is the worst case $\ell_\infty$-norm among $\{T e_k\}_{k=1}^n$ and $\{\widetilde{T} e_k\}_{k=1}^n$.
In certain scenarios, these $\ell_\infty$-norms are unevenly distributed.
For example, in compressed sensing with the analysis sparsity model,
$T$ is given by $T = \Phi \Psi^\dagger$
with the sensing transform $\Psi \in \cz^{n \times d}$ and the sparsifying transform $\Phi \in \cz^{N \times d}$.
If $\Psi$ and $\Phi$ correspond to the DFT and DWT (discrete wavelet transform), respectively, low-frequency atoms have larger correlations.
Thus there are a few $T e_k$s that dominate the others with large $\ell_\infty$-norms.
Krahmer and Ward \cite{krahmer2014stable} proposed a clever idea of sampling measurements
with respect to a variable density adapted to the local incoherence parameters,
which are $\{\norm{T e_k}_\infty\}_{k=1}^n$ and $\{\norm{\widetilde{T} e_k}_\infty\}_{k=1}^n$.\footnote{Krahmer and Ward \cite{krahmer2014stable} considered orthonormal $\Phi$ and $\Psi$ respectively corresponding to the DFT and the Haar DWT. In this case, $T^\dagger = T^*$. Thus, $\widetilde{T} = T$.}
Then the sample complexity depends on not the worst case incoherence parameter but the average of the local incoherence parameters.
In this section, adopting the idea by Krahmer and Ward \cite{krahmer2014stable},
we extend the results in Section~\ref{sec:main_result} to the case where the local incoherence parameters are unevenly distributed.
The following theorem is analogous to Theorem~\ref{thm:uniqueness} and provides a sufficient condition for recovery of sparse signals in a transform domain when measurements are sampled according to a variable density.

\begin{theorem}
\label{thm:uniqueness_vd}
Suppose $T: \cz^n \to \cz^N$ is injective.
Let $\mu_k$ and $\tilde{\mu}_k$ be defined by
\begin{equation}
\label{eq:loc_incoherenceT}
\mu_k = n \norm{\gamma^{1/2} T e_k}_\infty^2
\quad \mathrm{and} \quad
\tilde{\mu}_k = n \norm{\gamma^{-1/2} \widetilde{T} e_k}_\infty^2
\end{equation}
for all $k \in [n]$, where $\gamma$ is defined in \eqref{eq:defgamma} from $T$ and $\widetilde{T} = (T^\dagger)^*$.
Let $\Omega = \{\omega_1, \ldots, \omega_m\}$ be a multi-set of random indices
where $\omega_k$s are independent copies of a random variable $\omega$ with the following distribution:
\begin{equation}
\label{eq:variable_density}
\mathbb{P}(\omega = k) = \frac{\sqrt{\mu_k \tilde{\mu}_k}}{\sum_{j=1}^n {\sqrt{\mu_j \tilde{\mu}_j}}}, \quad \forall k \in [n].
\end{equation}
Suppose $T x$ is $s$-sparse.
Then with probability $1 - e^{-\beta} - 3/n$, $x$ is the unique minimizer to \eqref{eq:ell1min} provided
\[
m \geq \frac{C(1+\beta) \bar{\mu} s}{1 - \norm{\gamma T^* T - I_n}} \left[ \log N
+ \log \left( \norm{T}_{1 \to 2} \norm{T^\dagger}_{2 \to \infty} \right) \right],
\]
where $\bar{\mu}$ is defined as
\begin{equation}
\label{eq:defbarmu}
\bar{\mu} := \frac{1}{n} \sum_{k=1}^n \sqrt{\mu_k \tilde{\mu}_k}
\end{equation}
\end{theorem}

Compared to Theorem~\ref{thm:uniqueness}, Theorem~\ref{thm:uniqueness_vd} provides a performance guarantee at a smaller sample complexity,
where the worst-case incoherence parameter is replaced by the average incoherence parameter $\bar{\mu}$.
Note that $\bar{\mu}$ is always no greater than the worst-case incoherence parameter $(\max_k \mu_k)^{1/2} (\max_k \tilde{\mu}_k)^{1/2}$.
In particular when there exist dominant $\mu_k$s or $\tilde{\mu}_k$s compared to other incoherence parameters,
the sample complexity of Theorem~\ref{thm:uniqueness_vd} is much smaller than that of Theorem~\ref{thm:uniqueness}.

\begin{proof}[Proof of Theorem~\ref{thm:uniqueness_vd}]
Define
\begin{equation}
\label{eq:defnu}
\nu := [\sqrt{\mu_1}, \ldots, \sqrt{\mu_n}]^\transpose
\quad \mathrm{and} \quad
\tilde{\nu} := [\sqrt{\tilde{\mu}_1}, \ldots, \sqrt{\tilde{\mu}_n}]^\transpose.
\end{equation}
Without loss of generality, we may assume that $\mu_k$s and $\tilde{\mu}_k$s are strictly positive.
Then all entries of $\nu$ and $\tilde{\nu}$ are nonzero.

Using $\nu$ and $\tilde{\nu}$, we construct a pair of weighted transforms $W$ and $\widetilde{W}$ as
\begin{equation}
\label{eq:defW}
W = \sqrt{\bar{\mu}} T [\mbox{diag}(\nu)]^{-1}
\quad \mathrm{and} \quad
\widetilde{W} = \sqrt{\bar{\mu}} \widetilde{T} [\mbox{diag}(\tilde{\nu})]^{-1}.
\end{equation}
Then, $W$ and $\widetilde{W}$ satisfy
\begin{equation}
\label{eq:incoW}
\max_{k \in [n]} \norm{\gamma^{1/2} W e_k}_\infty \leq \sqrt{\frac{\bar{\mu}}{n}}
\quad \mathrm{and} \quad
\max_{k \in [n]} \norm{\gamma^{-1/2} \widetilde{W} e_k}_\infty \leq \sqrt{\frac{\bar{\mu}}{n}}.
\end{equation}
Furthermore, we have
\begin{equation}
\label{eq:isoW}
\begin{aligned}
\mathbb{E} \left( \frac{n}{m} \sum_{j=1}^m W S_\Omega^* S_\Omega \widetilde{W}^* \right)
{} & = \mathbb{E} \left( \frac{n}{m} \sum_{j=1}^m W e_{\omega_j} e_{\omega_j}^* \widetilde{W}^* \right) \\
{} & = \mathbb{E} \left( \frac{n}{m} \sum_{j=1}^m \bar{\mu} T [\mbox{diag}(\nu)]^{-1} e_{\omega_j} e_{\omega_j}^* [\mbox{diag}(\tilde{\nu})]^{-1} (\widetilde{T})^* \right) \\
{} & = \mathbb{E} \left( \frac{n}{m} \sum_{j=1}^m \frac{\bar{\mu}}{\sqrt{\mu_{\omega_j} \tilde{\mu}_{\omega_j}}} T e_{\omega_j} e_{\omega_j}^* T^\dagger \right) \\
{} & = \frac{1}{m} \sum_{j=1}^m T \widetilde{T}^\dagger = T T^\dagger.
\end{aligned}
\end{equation}
Since $\mbox{diag}(\tilde{\nu})$ is an invertible matrix,
$T$ and $W$ span the same subspace and $T T^\dagger$ is an orthogonal projection onto the span of $W$.
Furthermore,
\[
\langle W e_{k'}, \widetilde{W} e_k \rangle = 0, \quad \forall k \neq k'.
\]

Let $g = \sqrt{\bar{\mu}} [\mbox{diag}(\nu)]^{-1} g'$.
Then $T g = W g'$.
Thus \eqref{eq:ell1min} is equivalent to
\begin{equation}
\label{eq:ell1min_vd}
\minimize_{g' \in \cz^n} \, \norm{W g'}_1
\quad \mathrm{subject~to} \quad S_\Omega g' = \bar{\mu}^{-1/2} S_\Omega [\mbox{diag}(\nu)] x.
\end{equation}

Applying Theorem~\ref{thm:uniqueness} to \eqref{eq:ell1min_vd} with incoherence parameter $\bar{\mu}$ completes the proof.
\end{proof}

In the case when we are given sampled measurements corrupted with additive noise,
Theorems~\ref{thm:stability1} and \ref{thm:stability2} are modified
according to the change of the distribution for choosing random sample indices.

\begin{theorem}
\label{thm:stability1_vd}
Suppose the hypotheses of Theorem~\ref{thm:uniqueness_vd} hold.
Let $x^\sharp$ be a noisy version of $x$ that satisfies \eqref{eq:xharp_fid}.
Let $\hat{x}$ be the minimizer to
\begin{equation}
\label{eq:ell1min_noisy_vd}
\minimize_{g \in \cz^n} \, \norm{T g}_1
\quad \mathrm{subject~to} \quad \norm{S_{\Omega'} [\rho \odot (g - x^\sharp)]}_2 \leq \epsilon,
\end{equation}
where $\rho := \bar{\mu}^{-1/2} [\sqrt{\mu_1}, \dots, \sqrt{\mu_n}]^\transpose$.
Then
\[
\norm{\hat{x} - x}_2 \leq
\frac{\sigma_{\max}(T)}{\sigma_{\min}(T)}
\cdot
\frac{\bar{\mu}}{\min_{k \in [n]} \tilde{\mu}_k}
\cdot
\left[ 2 + 28 \sqrt{N} \left( \frac{\max_{k \in [n]} \mu_k}{\min_{k \in [n]} \tilde{\mu}_k} \cdot 3 n \norm{T}_{1 \to 2} \norm{T^\dagger}_{2 \to \infty} + 1 \right) \right] \epsilon.
\]
\end{theorem}

Like Theorem~\ref{thm:uniqueness_vd}, Theorem~\ref{thm:stability1_vd} provides a performance guarantee
at a lower sample complexity (in order) when compared to Theorem~\ref{thm:stability1}.
On the other hand, the noise amplification factor of Theorem~\ref{thm:stability1_vd} is larger than that of Theorem~\ref{thm:stability1}
by a factor that depends on the distribution of the local incoherence parameters $\{(\mu_k,\tilde{\mu}_k)\}_{k=1}^n$.

\begin{proof}[Proof of Theorem~\ref{thm:stability1_vd}]
Let $W$ and $\widetilde{W}$ be defined in \eqref{eq:defW}.
Let $\nu$ and $\tilde{\nu}$ be defined in \eqref{eq:defnu}.
Define
\begin{equation}
\label{eq:defLambda}
\Lambda := \sqrt{\bar{\mu}} [\mbox{diag}(\nu)]^{-1}
\quad \mathrm{and} \quad
\widetilde{\Lambda} = \sqrt{\bar{\mu}} [\mbox{diag}(\tilde{\nu})]^{-1}.
\end{equation}
Then $W = T \Lambda$ and \eqref{eq:ell1min_noisy} is equivalent to
\begin{equation}
\label{eq:ell1min_noisy_vd2}
\minimize_{g' \in \cz^n} \, \norm{W g'}_1
\quad \mathrm{subject~to} \quad \norm{S_{\Omega'} g' - S_{\Omega'} \Lambda^{-1} x^\sharp}_2 \leq \epsilon.
\end{equation}

Let $\breve{x}$ denote the minimizer to \eqref{eq:ell1min_noisy_vd2}.
Then, the minimizer $\hat{x}$ to \eqref{eq:ell1min_noisy} is represented as $\hat{x} = \Lambda \breve{x}$,
i.e. $\breve{x} = \Lambda^{-1} \hat{x}$.
By applying Theorem~\ref{thm:stability1} to \eqref{eq:ell1min_noisy_vd2}, we obtain
\begin{equation}
\label{eq:pf_thm_stability1_vd:bnd}
\begin{aligned}
\norm{T \hat{x} - T x}_2
{} & = \norm{W \breve{x} - W \Lambda^{-1} x}_2 \\
{} & \leq \left\{ 2 + 28 \sqrt{N} \left( 3 n \norm{T \Lambda}_{1 \to 2} \norm{\widetilde{\Lambda}^{-1} T^\dagger}_{2 \to \infty} + 1 \right) \right\} \epsilon \norm{W}.
\end{aligned}
\end{equation}

The proof completes by applying the following inequalities to \eqref{eq:pf_thm_stability1_vd:bnd}:
\[
\norm{\widetilde{\Lambda}^{-1} T^\dagger}_{2 \to \infty} \leq \norm{\widetilde{\Lambda}^{-1}} \norm{T^\dagger}_{2 \to \infty}
= \frac{\bar{\mu} \norm{T^\dagger}_{2 \to \infty}}{\min_{k \in [n]} \tilde{\mu}_k}
\]
and
\[
\norm{T \Lambda}_{1 \to 2} \leq \norm{T}_{1 \to 2} \norm{\Lambda}
= \frac{\max_{k \in [n]} \mu_k \norm{T^\dagger}_{2 \to \infty}}{\bar{\mu}}.
\]
\end{proof}

\begin{theorem}
\label{thm:stability2_vd}
Suppose $T: \cz^n \to \cz^N$ is injective.
Suppose that $T$ and $\widetilde{T} = (T^\dagger)^*$ satisfy \eqref{eq:loc_incoherenceT}
with parameters $\{\mu_k\}_{k=1}^n$, $\{\tilde{\mu}_k\}_{k=1}^n$, and $\gamma$.
Let $\bar{\mu}$ be defined in \eqref{eq:defbarmu}.
Let $\Omega = \{\omega_1, \ldots, \omega_m\}$ be a multi-set of random indices
where $\omega_k$s are i.i.d. copies of a random variable $\omega$ with the distribution in \eqref{eq:variable_density}.
Let $\hat{x}$ be the minimizer to \eqref{eq:ell1min_noisy_vd} with $x^\sharp$ satisfying \eqref{eq:xharp_fid}.
Then there exist numerical constants $C, c > 0$ for which the following holds.
With probability $1 - N^{-4}$, we have
\[
\norm{\hat{x} - x}_2 \leq
\frac{\sigma_{\max}(T)}{\sigma_{\min}(T)}
\cdot
\frac{\bar{\mu}}{\min_{k \in [n]} \tilde{\mu}_k}
\cdot
\left[ 14\sqrt{N} + \frac{n}{m} \left\{ \frac{\max_{k \in [n]} \mu_k}{\min_{k \in [n]} \tilde{\mu}_k} \cdot \norm{T}_{1 \to 2} \norm{T^\dagger}_{2 \to \infty} (|\Omega|-|\Omega'|) +1 \right\} \right] \epsilon,
\]
provided
\[
m \geq \frac{C \bar{\mu} s \log^4 N}{1 - \norm{\gamma T^* T - I_n}}.
\]
Furthermore, if $T^* T = I_n$, then with probability $1 - N^{-4}$,
\[
\norm{\hat{x} - x}_2 \leq
\frac{\sigma_{\max}(T)}{\sigma_{\min}(T)}
\cdot
\frac{\bar{\mu}}{\min_{k \in [n]} \tilde{\mu}_k}
\cdot
\left( 14\sqrt{N} + \frac{n}{m} \cdot \frac{\max_{k \in [n]} \mu_k}{\min_{k \in [n]} \tilde{\mu}_k} \cdot R \right) \epsilon,
\]
where $R$ is the count of the most repeated elements in $\Omega$.
\end{theorem}

\begin{proof}[Proof of Theorem~\ref{thm:stability1_vd}]
Let $W$ and $\widetilde{W}$ be defined in \eqref{eq:defW}.
Let $\nu$ and $\tilde{\nu}$ be defined in \eqref{eq:defnu}.
Let $\Lambda$ and $\widetilde{\Lambda}$ be defined in \eqref{eq:defLambda}.
In the proof of Theorem~\ref{thm:stability1_vd},
we have shown that \eqref{eq:ell1min_noisy_vd} is equivalent to \eqref{eq:ell1min_noisy_vd2}.
In the proof of Theorem~\ref{thm:uniqueness_vd},
we have shown that $W$ and $\widetilde{W}$ satisfy \cref{eq:incoW,eq:isoW}.
Therefore, applying Theorem~\ref{thm:stability2} to \eqref{eq:ell1min_vd} with incoherence parameter $\bar{\mu}$ completes the proof.
\end{proof}

\section{Extension to Group Sparsity Models}
\label{sec:main_result_gs}

In Sections~\ref{sec:main_result} and \ref{sec:main_result_vd}, we considered the sparsity model in a transform domain. In applications, the transform $T x$ exhibits additional structures -- group sparsity. For example, in compressed sensing of 2D signals, the sparsifying transform $\Phi = [\Phi_1^\top, \Phi_2^\top]^\top$ can be the concatenation of the horizontal and vertical finite difference operators $\Phi_1$ and $\Phi_2$. Anisotropic total variation encourages the sparsity of $\Phi f$ \cite{krahmer2014stable}. On the other hand, one can choose isotropic total variation, assuming that $\Phi_1 f$ and $\Phi_2 f$ are jointly sparse (see the experiments section). For another example, in compressed sensing of color images or hyperspectral images, the sparse codes acquired from applying the analysis operator to different channels are usually assumed to be jointly sparse. To exploit such structures, we extend the results in the previous sections to group sparsity models in a transform domain.
More specifically, we assume that the transform $T x$ of unknown signal $x$ via $T: \cz^n \to \cz^L$ is $(s,t)$ strongly group sparse in the following sense.

\begin{defi}[{\cite{huang2010benefit}}]
Let $\calG = \{\calG_1, \ldots, \calG_N\}$ be a partition of $[L]$, i.e. $\bigcup_{j \in [N]} \calG_j = [L]$ and $\calG_j \bigcap \calG_{j'} = \emptyset$ for $j \neq j'$. A vector $z \in \cz^L$ is $(s,t)$ strongly group sparse with respect to $\calG$ if there exists $J \subset [N]$ such that
\[
\supp{z} \subset \calG_J,
\quad
|\calG_J| \leq t,
\quad \mathrm{and} \quad
|J| \leq s,
\]
where $\calG_J$ is defined by
\[
\calG_J := \bigcup_{j \in J} \calG_j.
\]
\end{defi}

An atomic norm for this group sparsity model is given by
\[
\tnorm{z}_{\calG,1} := \sum_{j \in [N]} \norm{\Pi_{\calG_j} z}_2.
\]
The dual norm of $\tnorm{\cdot}_{\calG,1}$ is defined by
\[
\tnorm{z}_{\calG,\infty} := \sup_{\zeta \in \cz^L: \tnorm{\zeta}_{\calG,1} \leq 1} |\langle \zeta, z \rangle|,
\]
which is equivalently rewritten as
\[
\tnorm{z}_{\calG,\infty} = \max_{j \in [N]} \norm{\Pi_{\calG_j} z}_2.
\]

A subgradient $\tilde{z}$ of $\tnorm{\cdot}_{\calG,1}$ at $z$ is given by
\begin{equation}
\label{eq:subgrad_mixed}
\Pi_{\calG_j} \tilde{z}
=
\begin{dcases}
\frac{\Pi_{\calG_j} z}{\norm{\Pi_{\calG_j} z}_2}, & \Pi_{\calG_j} z \neq 0, \\
0, & \mathrm{otherwise},
\end{dcases}
\qquad \forall j \in [N].
\end{equation}

In a special case where $\calG_j = \{j\}$ for all $j \in [N]$,
the strong group sparsity level reduces to the conventional sparsity model.
The analogy between the two models is summarized in Table~\ref{tab:replace_notation}.
All the results in the previous section generalize to the strong group sparsity model according to this analogy.
In the below, we state the extended results as Theorems, the proofs of which are obtained in a straightforward way
by modifying the proofs of analogous theorems and lemmas for the usual sparsity model according to Table~\ref{tab:replace_notation}.
Thus, we do not repeat the proofs in this section.
(The only exception is Lemma~\ref{lemma:E3} and we provide an analogous Lemma~\ref{lemma:E3'} in Section~\ref{subsec:riplike_lemmas}.)

\begin{table}
\newcolumntype{C}{>{\centering\arraybackslash}p{30ex}}
\caption{Analogy between Sparsity Model and Group Sparsity Model}
\label{tab:replace_notation}
\begin{center}
\begin{tabular}{>{\centering\arraybackslash}p{20ex}|C|C}
& Sparsity Model & Group Sparsity Model \\\hline\hline
$\mbox{dim}(R(T))$ & $N$ & $L = \sum_{j=1}^N |\calG_j|$ \\\hline
\multirow{2}{*}{support} & \multirow{2}{*}{$\supp{z} = \{ j :~ \Pi_{\{j\}} z \neq 0 \}$} & $\gsupp{z} = \{ j :~ \Pi_{\calG_j} z \neq 0 \}$ \\
& & $\supp{z} = \bigcup_{j \in \gsupp{z}} \calG_j$ \\\hline
group sparsity level & $\norm{z}_0 = |\supp{z}|$ & $\tnorm{z}_{\calG,0} = |\gsupp{z}|$ \\\hline
total sparsity level & $\norm{z}_0 = |\supp{z}|$ & $\norm{z}_0 = |\supp{z}|$ \\\hline
atomic norm & $\norm{z}_1 = \sum_{j=1}^N |\Pi_{\{j\}} z|$ & $\tnorm{\cdot}_{\calG,1} = \sum_{j=1}^N \norm{\Pi_{\calG_j} z}_2$ \\\hline
dual norm & $\norm{\cdot}_\infty = \max_{j \in [N]} |\Pi_{\{j\}} z|$ & $\tnorm{\cdot}_{\calG,\infty} = \max_{j \in [N]} \norm{\Pi_{\calG_j} z}_2$ \\\hline
subgradient & \multirow{2}{*}{\eqref{eq:sgn}} & \multirow{2}{*}{\eqref{eq:subgrad_mixed}} \\
of atomic norm & & \\\hline
incoherence & \multirow{2}{*}{$\mu$} & \multirow{2}{*}{$\mu_\calG$} \\
parameter & & \\\hline
local incoherence & \multirow{2}{*}{$\{\mu_k\}_{k=1}^n$} & \multirow{2}{*}{$\{\mu_{\calG,k}\}_{k=1}^n$} \\
parameters & & \\\hline
\end{tabular}
\end{center}
\end{table}

When $T x$ is strongly group sparse,
we propose to estimate $T x$ by solving the following optimization problem, which generalizes \eqref{eq:ell1min}:
\begin{equation}
\label{eq:mixedmin}
\minimize_{g \in \cz^n} \, \tnorm{T g}_{\calG,1}
\quad \mathrm{subject~to} \quad S_\Omega g = S_\Omega x.
\end{equation}
Then the following theorem, which is analogous to Theorem~\ref{thm:uniqueness}, provides a performance guarantee for \eqref{eq:mixedmin}.

\begin{theorem}
\label{thm:uniqueness_gs}
Let $\calG = \{\calG_1, \ldots, \calG_N\}$ be a partition of $[L]$.
Suppose $T: \cz^n \to \cz^L$ is injective.
Let $\mu_\calG$ be given by
\begin{equation}
\label{eq:incoherence_gs}
\mu_\calG = \max_{k \in [n]} \max_{j \in [N]}
\max\left\{
n \norm{\gamma^{1/2} \Pi_{\calG_j} T e_k}_2^2,
n \norm{\gamma^{-1/2} \Pi_{\calG_j} \widetilde{T} e_k}_2^2
\right\},
\end{equation}
where $\gamma$ is defined in \eqref{eq:defgamma} from $T$ and $\widetilde{T} = (T^\dagger)^*$.
Let $\Omega = \{\omega_1,\ldots,\omega_m\}$ be a multi-set of random indices
where $\omega_k$s are i.i.d. following the uniform distribution on $[n]$.
Suppose $T x$ is $(s,t)$ strongly group sparse with respect $\calG$.
Then with probability $1 - e^{-\beta} - 3/n$, $x$ is the unique minimizer to \eqref{eq:mixedmin} provided
\begin{equation}
\label{eq:samprate_gs}
m \geq \frac{C(1+\beta) \mu_\calG s}{1 - \norm{\gamma T^* T - I_n}} \left[ \log N
+ \log \left( \norm{T}_{1 \to 2} \norm{T^\dagger}_{2 \to \infty} \right) \right].
\end{equation}
\end{theorem}

Suppose that $\mu$ and $\mu_\calG$ are the smallest constants satisfying corresponding incoherence conditions.
If $t = \ell s$ and $|\calG_j| = \ell$ for all $j \in [N]$, then
\[
\norm{\Pi_{\calG_j} T e_k}_2 \leq \sqrt{|\calG_j|} \norm{T e_k}_\infty, \quad \forall j \in [N], ~ \forall k \in [n].
\]
Thus, we have $\mu_\calG \leq \ell \mu$ and a sufficient condition for \eqref{eq:samprate_gs} is given by
\[
m \geq \frac{C(1+\beta) \mu t}{1 - \norm{\gamma T^* T - I_n}} \left[ \log N
+ \log \left( \norm{T}_{1 \to 2} \norm{T^\dagger}_{2 \to \infty} \right) \right].
\]
In other words, the sample complexity is proportional to the total sparsity level $t$ and there is no gain from the group structure.
This inequality is tight if each $\Pi_{\calG_j} T e_k$ have nonzero elements of the same magnitude.
Contrarily, if nonzero elements of each $\Pi_{\calG_j} T e_k$ vary a lot in their magnitudes,
$\mu_\calG$ is smaller than $\ell \mu$ and there is gain from the group sparsity structure.

In the presence of noise to measurements,
we generalize the optimization formulation for recovery in \eqref{eq:ell1min_noisy} as follows:
\begin{equation}
\label{eq:mixedmin_noisy}
\minimize_{g \in \cz^n} \, \tnorm{T g}_{\calG,1}
\quad \mathrm{subject~to} \quad \norm{S_{\Omega'} g - S_{\Omega'} x^\sharp}_2 \leq \epsilon.
\end{equation}

The following theorem, analogous to Theorem~\ref{thm:stability1}, provides a performance guarantee for \eqref{eq:mixedmin_noisy}.

\begin{theorem}
\label{thm:stability1_gs}
Suppose the hypotheses of Theorem~\ref{thm:uniqueness_gs} hold.
Let $\hat{x}$ be the minimizer to \eqref{eq:mixedmin_noisy} with $x^\sharp$ satisfying
\[
\norm{S_{\Omega'} (x - x^\sharp)}_2 \leq \epsilon.
\]
Then
\[
\norm{\hat{x} - x}_2 \leq
\frac{\sigma_{\max}(T)}{\sigma_{\min}(T)}
\cdot
\left\{ 2 + 28 \sqrt{N} \left( 3 n \norm{T}_{1 \to 2} \norm{T^\dagger}_{2 \to \infty} + 1 \right) \right\} \epsilon.
\]
\end{theorem}

The results for recovery using a variable density sampling designed from local incoherence parameters generalize in a similar way.
We state the results in the following theorems.

\begin{theorem}
\label{thm:uniqueness_vd_gs}
Let $\calG = \{\calG_1, \ldots, \calG_N\}$ be a partition of $[L]$.
Suppose $T: \cz^n \to \cz^N$ is injective.
Let $\mu_{\calG,k}$ and $\tilde{\mu}_{\calG,k}$ be given by
\begin{equation}
\label{eq:loc_incoherenceT_gs}
\mu_{\calG,k} = \max_{j \in [N]} n \norm{\gamma^{1/2} \Pi_{\calG_j} T e_k}_2^2
\quad \mathrm{and} \quad
\tilde{\mu}_{\calG,k} = \max_{j \in [N]} n \norm{\gamma^{-1/2} \Pi_{\calG_j} \widetilde{T} e_k}_2^2,
\end{equation}
where $\gamma$ is defined in \eqref{eq:defgamma} from $T$ and $\widetilde{T} = (T^\dagger)^*$.
Let $\Omega = \{\omega_1, \ldots, \omega_m\}$ be a multi-set of random indices
where $\omega_k$s are independent copies of a random variable $\omega$ with the following distribution:
\begin{equation}
\label{eq:variable_density_gs}
\mathbb{P}(\omega = k) = \frac{\sqrt{\mu_{\calG,k} \tilde{\mu}_{\calG,k}}}{\sum_{j=1}^n \sqrt{\mu_{\calG,j} \tilde{\mu}_{\calG,j}}}, \quad \forall k \in [n].
\end{equation}
Suppose $T x$ is $(s,t)$ strongly group sparse with respect $\calG$.
Then with probability $1 - e^{-\beta} - 3/n$, $x$ is the unique minimizer to \eqref{eq:ell1min} provided
\[
m \geq \frac{C(1+\beta) \bar{\mu}_\calG s}{1 - \norm{\gamma T^* T - I_n}} \left[ \log N
+ \log \left( \norm{T}_{1 \to 2} \norm{T^\dagger}_{2 \to \infty} \right) \right],
\]
where $\bar{\mu}$ is defined as
\[
\bar{\mu}_\calG := \frac{1}{n} \sum_{k=1}^n \sqrt{\mu_{\calG,k} \tilde{\mu}_{\calG,k}}.
\]
\end{theorem}

\begin{theorem}
\label{thm:stability1_vd_gs}
Suppose the hypotheses of Theorem~\ref{thm:uniqueness_vd_gs} hold.
Let $\hat{x}$ be the minimizer to \eqref{eq:mixedmin_noisy} with $x^\sharp$ satisfying
\begin{equation}
\label{eq:xharp_fid_gs}
\norm{S_{\Omega'} x - S_{\Omega'} x^\sharp}_2 \leq \epsilon.
\end{equation}
Then
\[
\norm{\hat{x} - x}_2 \leq
\frac{\sigma_{\max}(T)}{\sigma_{\min}(T)}
\cdot
\frac{\bar{\mu}_\calG}{\min_{k \in [n]} \tilde{\mu}_{\calG,k}}
\cdot
\left[ 2 + 28 \sqrt{N} \left( \frac{\max_{k \in [n]} \mu_{\calG,k}}{\min_{k \in [n]} \tilde{\mu}_{\calG,k}} \cdot 3 n \norm{T}_{1 \to 2} \norm{T^\dagger}_{2 \to \infty} + 1 \right) \right] \epsilon.
\]
\end{theorem}

\section{Circulant Transforms}
\label{sec:circulant}

Many sparsifying transforms fall into the category of circulant transforms, including the identity transform. A block transform (e.g., block DCT), when applied to all overlapping patches of a signal (sliding window with stride $1$, including the wrap-around patches at the edges), is a union of circulant transforms applied to the signal \cite{pfister2015learning}.
In this section, we consider only the case when the measurement matrix $\Psi\in\cz^{n\times n}$ is the DFT matrix, and compute the variable density sampling distribution \eqref{eq:variable_density} for sparsity with respect to a circulant transform, and distribution \eqref{eq:variable_density_gs} for joint sparsity with respect to a union of circulant transforms. We show that, if the circulant transforms are injective, the distributions \eqref{eq:variable_density} and \eqref{eq:variable_density_gs} correspond to the uniform distribution on $[n]$. On the other hand, some circulant transforms are not injective (e.g., the finite difference operator for 1D total variation or 2D isotropic total variation). Even in this case, we show that the ``variable'' density sampling distributions are a variation of the uniform distribution.

\subsection{Injective Circulant Transforms}
We say $\Phi\in \cz^{n\times n}$ is a circulant transform (circulant matrix), if $\Phi f = \phi \circledast f$ is the circular convolution of $f$ with some vector $\phi \in \cz^n$, i.e. the matrix representation of $\Phi$ is given by
\begin{equation}
\label{eq:circulant}
\Phi = \begin{bmatrix}
\phi[1] & \phi[n] & \phi[n-1] & \cdots & \phi[2] \\
\phi[2] & \phi[1] & \phi[n] & \cdots & \phi[3] \\
\phi[3] & \phi[2] & \phi[1] & \cdots & \phi[4] \\
\vdots & \vdots & \vdots & \ddots & \vdots \\
\phi[n] & \phi[n-1] & \phi[n-2] & \cdots & \phi[1] \\
\end{bmatrix}
\end{equation}
The identity transform is a circulant transform with $\phi=e_1$.

By linearity and shift invariance of circular convolution, a circulant transform $\Phi$ can always be diagonalized by the DFT matrix $\Psi$:
\begin{equation}
\label{eq:diagonalize}
\Phi = \Psi^* \diag(\lambda) \Psi.
\end{equation}
where $\lambda = \sqrt{n}\Psi \phi$ is the (unnormalized) DFT of $\phi$. The same argument also applies to a 2D circulant transform (circulant block circulant matrix) and the 2D DFT matrix. To avoid verbosity, we use circulant transform $\Phi$ and DFT matrix $\Psi$ to denote both 1D and 2D transforms.

\begin{theorem}\label{thm:circulant1}
For the DFT matrix $\Psi\in\cz^{n\times n}$ and an invertible circulant transform $\Phi\in\cz^{n\times n}$, the sampling density distribution \eqref{eq:variable_density} is the uniform distribution on $[n]$.
\end{theorem}

\begin{proof}[Proof of Theorem~\ref{thm:circulant1}]
By the diagonalization in \eqref{eq:diagonalize}, we have
\begin{align*}
& T = \Phi \Psi^\dagger = \Psi^* \diag(\lambda),\\
& \widetilde{T} = (T^\dagger)^* = \Psi^* \diag(\tilde{\lambda}),
\end{align*}
where $\tilde{\lambda}$ is the complex conjugate of the element-wise inverse of $\lambda$, i.e. $\tilde{\lambda}[k]=(\lambda[k]^*)^{-1}$. Then, the following choice of parameters $\mu_k$ and $\tilde{\mu}_k$ satisfy \eqref{eq:loc_incoherenceT}:
\begin{align*}
& \mu_k = n\gamma \norm{Te_k}_\infty^2 = n\gamma |\lambda[k]|^2 \norm{\Psi^*e_k}_\infty^2 = \gamma |\lambda[k]|^2,\\
& \tilde{\mu}_k = n\gamma^{-1} \norm{\widetilde{T}e_k}_\infty^2 = n\gamma^{-1} |\tilde{\lambda}[k]|^2 \norm{\Psi^*e_k}_\infty^2 = \gamma^{-1} |\lambda[k]|^{-2},
\end{align*}
where $\Psi^*e_k$ is the $k$th column of the discrete Fourier basis and has infinity norm $1/\sqrt{n}$. Therefore, $\mu_k\tilde{\mu}_k = 1$ for all $k\in [n]$, and the distribution in \eqref{eq:variable_density} is $\mathbb{P}(\omega = k) = 1/n$, i.e. the uniform distribution on $[n]$.
\end{proof}

Applying a sparsifying circulant transform to a signal is equivalent to passing the signal through a sparsifying filter. One may also pass the signal through a bank of sparsifying filters.
The filter bank is equivalent to a union of circulant transforms, concatenated as
\begin{equation}
\label{eq:concatenation}
\Phi = [\Phi_1^\transpose, \Phi_2^\transpose,\dots, \Phi_\ell^\transpose]^\transpose \in\cz^{L\times n},
\end{equation}
where $L=\ell n$.
For example, the patch transform sparsity model \cite{pfister2015learning} corresponds to this case. The 2D isotropic total variation model (see Section \ref{sec:TV}), as another example, has an additional joint sparsity structure. For the latter example, let us consider a particular partition $\calG = \{\calG_1, \ldots, \calG_n\}$ given by
\begin{equation}
\label{eq:partition}
\calG_k = \{(j-1)n+k\}_{j=1}^\ell, \qquad \forall k\in [n].
\end{equation}
For group sparsity on this union of circulant transforms, we have a result similar to Theorem \ref{thm:circulant1}.

\begin{theorem}\label{thm:circulant2}
For the DFT matrix $\Psi\in\cz^{n\times n}$, an injective transform $\Phi\in\cz^{L\times n}$ defined in \eqref{eq:concatenation}, and the partition defined in \eqref{eq:partition}, the sampling density distribution \eqref{eq:variable_density_gs} is the uniform distribution on $[n]$.
\end{theorem}

\begin{proof}[Proof of Theorem~\ref{thm:circulant2}]
Similar to \eqref{eq:diagonalize}, we have the following factorization for the concatenated transform:
\begin{equation}
\label{eq:diagonalize2}
\Phi = (I_\ell \otimes \Psi^*) [\diag(\lambda_1),\diag(\lambda_2),\dots, \diag(\lambda_\ell)]^\top \Psi,
\end{equation}
where $\lambda_j$ is the (unnormalized) DFT of $\phi_j$, the convolution kernel of the $j$th cirulant transform $\Phi_j$. Hence $T$ and $\widetilde{T}$ are
\begin{align*}
& T = \Phi \Psi^\dagger = (I_\ell \otimes \Psi^*)[\diag(\lambda_1),\diag(\lambda_2),\dots, \diag(\lambda_\ell)]^\top,\\
& \widetilde{T} = (T^\dagger)^* = T(T^*T)^{-1} =  (I_\ell \otimes \Psi^*) [\diag(\tilde{\lambda}_1),\diag(\tilde{\lambda}_2),\dots, \diag(\tilde{\lambda}_\ell)]^\top,
\end{align*}
where $\tilde{\lambda}_j[k] = \lambda_j[k]/\bigl(\sum_{j'=1}^{\ell} |\lambda_{j'}[k]|^2\bigr)$. Then, $\mu_{\calG,k}$ and $\tilde{\mu}_{\calG,k}$ in \eqref{eq:loc_incoherenceT_gs} are given respectively by
\begin{align*}
\mu_{\calG,k} & = n\gamma \max_{k'\in[n]}\norm{\Pi_{\calG_{k'}}Te_k}_2^2 \\
& = n\gamma \max_{k'\in[n]} \norm{\Pi_{\calG_{k'}}\bigl\{\bigl[\lambda_1[k],\lambda_2[k],\dots,\lambda_\ell[k]\bigr]^\transpose \otimes (\Psi^*e_k) \bigl\} }_2^2 \\
& = n\gamma \max_{k'\in[n]} |e_{k'}^\transpose \Psi^*e_k| \norm{\bigl[\lambda_1[k],\lambda_2[k],\dots,\lambda_\ell[k]\bigr]^\transpose}_2^2
= \gamma \bigl(\sum_{j=1}^{\ell} |\lambda_{j}[k]|^2\bigr)
\end{align*}
and
\begin{align*}
\tilde{\mu}_{\calG,k} & = n\gamma^{-1} \max_{k'\in[n]}\norm{\Pi_{\calG_{k'}}\widetilde{T}e_k}_2^2 \\
& = n\gamma^{-1} \max_{k'\in[n]} \norm{\Pi_{\calG_{k'}}\bigl\{\bigl[\tilde{\lambda}_1[k],\tilde{\lambda}_2[k],\dots,\tilde{\lambda}_\ell[k]\bigr]^\transpose \otimes (\Psi^*e_k) \bigl\} }_2^2 \\
& = \gamma^{-1} \bigl(\sum_{j=1}^{\ell} |\tilde{\lambda}_{j}[k]|^2\bigr)
= \gamma^{-1}  \bigl(\sum_{j=1}^{\ell} |\lambda_{j}[k]|^2\bigr)^{-1}.
\end{align*}
Therefore, $\mu_{\calG,k}\tilde{\mu}_{\calG,k} = 1$ for all $k\in [n]$, and the distribution in \eqref{eq:variable_density_gs} is $\mathbb{P}(\omega = k) = 1/n$, i.e. the uniform distribution on $[n]$.
\end{proof}

\subsection{Non-injective Circulant Transforms}\label{sec:TV}
In this section, we consider non-injective circulant transforms, such as the finite difference operator for 1D total variation, or the union of vertical and horizontal finite difference operators for 2D isotropic total variation. Since the spectral responses of these transforms are zero at certain frequencies, they are invariant to changes in the corresponding frequency components. Therefore, unless these frequency components are sampled in the measurement, they cannot be recovered from $\ell_1$-norm minimization \eqref{eq:ell1min} or mixed norm minimization \eqref{eq:mixedmin}. Assuming without loss of generality that the null frequencies are the first $n_0$ columns in $\Psi^*$ ($n_0<\min\{n,m\}$), we adopt the following two-step sampling scheme:
\begin{enumerate}
	\item Always sample indices $\omega_k = k$ for $k\in [n_0]$.
	\item Generate a multi-set of $m-n_0$ indices $\{\omega_{n_0+1},\omega_{n_0+2},\dots,\omega_m\}$, which are i.i.d. following a distribution on $\{n_0+1,n_0+2,\dots,n\}$.
\end{enumerate}
We can compute the sampling density on $\{n_0+1,n_0+2,\dots,n\}$ based on \eqref{eq:loc_incoherenceT} and \eqref{eq:variable_density} (or \eqref{eq:loc_incoherenceT_gs} and \eqref{eq:variable_density_gs}) by removing the zero columns in $T=\Phi\Psi^*$ (replacing $T$ with $T[e_{n_0+1},e_{n_0+2},\dots,e_{n}]$). Next, we state variations of Theorems \ref{thm:circulant1} and \ref{thm:circulant2} in these cases.

\begin{cor} \label{cor:circulant1}
Suppose circulant transform $\Phi$ in \eqref{eq:diagonalize} satisfies $\lambda[k] = 0$ for $k\in[n_0]$. Then the sampling density in Step 2), for $\ell_1$-norm minimization \eqref{eq:ell1min} based on \eqref{eq:variable_density}, is the uniform distribution on $\{n_0+1,n_0+2,\dots,n\}$.
\end{cor}

\begin{cor} \label{cor:circulant2}
Suppose concatenated transform $\Phi$ in \eqref{eq:diagonalize2} satisfies $\lambda_j[k] = 0$ for all $j\in[\ell]$ and $k\in[n_0]$. Then the sampling density in Step 2), for mixed norm minimization \eqref{eq:mixedmin} based on \eqref{eq:variable_density_gs}, is the uniform distribution on $\{n_0+1,n_0+2,\dots,n\}$.
\end{cor}

These results are direct consequences of Theorems \ref{thm:circulant1} and \ref{thm:circulant2}, whose proofs translate with no changes other than removing the zero columns in $T$ (the columns indexed by $[n_0]$).

Next, we specialize these results to total variation minimization.
We define the 1D finite difference operator $\Phi_{\mathrm{TV},n}$ by \eqref{eq:circulant}, where
\[
\phi[k] = \begin{cases}
1 & k = 1\\
-1 & k = 2\\
0 & k\in\{3,4,\dots, n\}.
\end{cases}
\]
Then $\norm{f}_{\mathrm{TV}} = \norm{\Phi_{\mathrm{TV},n} f}_1$ is the 1D total variation of $f$.
Clearly, the circulant transform $\Phi_{\mathrm{TV},n}$ is not injective, since its null space contains the direct current (DC) component -- $\Psi^*e_1$.
By Corollary \ref{cor:circulant1}, we adopt the following sampling scheme for total variation minimization:
\begin{enumerate}
	\item Always sample index $\omega_1 = 1$.
	\item Generate a multi-set of $m-1$ indices $\{\omega_2,\omega_3,\dots,\omega_m\}$, which are i.i.d. following the uniform distribution on $\{2,3,\dots,n\}$.
\end{enumerate}


Total variation is more commonly used for 2D signals (e.g., images). For a 2D signal $f$ of size $n_1\times n_2$, the finite difference operator is a concatenation of the vertical and horizontal finite difference operators:
\[
\Phi_{\mathrm{TV},n_1,n_2} = \begin{bmatrix}
I_{n_2} \otimes \Phi_{\mathrm{TV},n_1} \\
\Phi_{\mathrm{TV},n_2} \otimes I_{n_1}
\end{bmatrix}.
\]
The anisotropic and isotropic total variations of $f$ are defined by
\begin{align*}
& \norm{f}_{\mathrm{TV, aniso}} = \norm{\Phi_{\mathrm{TV},n_1,n_2} f}_1,\\
& \norm{f}_{\mathrm{TV, iso}} = \tnorm{\Phi_{\mathrm{TV},n_1,n_2} f}_{\calG,1},
\end{align*}
where the partion $\calG=\{\calG_1,\dots,\calG_n\}$ is defined by \eqref{eq:partition} for $\ell=2$ and $n=n_1n_2$. Let the measurement operator $\Psi$ be the 2D DFT on signals of size $n_1\times n_2$. Similar to the 1D case, the DC component $\Psi^* e_1$ belongs to the null space of $\Phi_{\mathrm{TV},n_1,n_2}$. By Corollary \ref{cor:circulant2}, we use the same two-step sampling scheme as in the 1D case.


\section{Proof of Theorem~\ref{thm:uniqueness}}
\label{sec:pf_main_result_noiseless}

In this section, we prove Theorem~\ref{thm:uniqueness},
which provides a sufficient condition for exact recovery of sparse signals in a transform domain from noiseless observations.
The proof of Theorem~\ref{thm:uniqueness} is based on the golfing scheme \cite{gross2011recovering},
which was originally proposed for matrix completion \cite{gross2011recovering}
and later adopted to compressed sensing \cite{candes2011probabilistic} and to structured matrix completion \cite{chen2014robust}.

The golfing scheme constructs an inexact dual certificate.
The notion of a dual certificate was originally proposed for compressed sensing (cf. \cite{candes2006near}).
To reconstruct an $s$-sparse $f \in \cz^d$ from $b = A f$, it was proposed to estimate $f$ as the solution to
\[
\minimize_{\tilde{f} \in \cz^d} \, \norm{\tilde{f}}_1 \quad \mathrm{subject~to} \quad b = A \tilde{f}.
\]
A dual certificate is a subgradient $v \in \cz^d$ of $\norm{\cdot}_1$ at $f$
such that $f$ is orthogonal to all null vectors of the sensing matrix $A$.
In matrix completion, the objective function is replaced from the $\ell_1$-norm to the nuclear norm
and the sensing matrix is replaced by a pointwise sampling operator.
Gross \cite{gross2011recovering} proposed the clever golfing scheme
that constructs an inexact dual certificate, which is close to the exact dual certificate,
and showed a low-rank matrix is exactly reconstructed from partial entries sampled at a near optimal rate.
Candes and Plan \cite{candes2011probabilistic} adopted the golfing scheme back to compressed sensing
and showed that exact recovery is guaranteed from $m = O(s \log d)$ incoherent measurements,
which improves on the previous performance guarantee with $m = O(s \log^4 d)$.

Chen and Chi \cite{chen2014robust} adopted the golfing scheme to structured matrix completion.
Let $T: \cz^n \to \cz^{n_1 \times n_1}$ be a linear operator that maps a vector $x \in \cz^n$ to
a structured matrix $T x \in \cz^{n_1 \times n_2}$ (e.g., a Hankel matrix).
Chen and Chi \cite{chen2014robust} proposed to estimate $x$ by
\[
\minimize_{g \in \cz^n} \, \norm{T g}_* \quad \mathrm{subject~to} \quad S_\Omega g = S_\Omega x.
\]
This can be interpreted as recovery of low-rank matrices in a special transform domain,
where $T$ maps the standard basis vectors $\{e_1,\ldots,e_n\}$ to unit-norm matrices with disjoint supports.
Unlike compressed sensing or matrix completion, in structured matrix completion,
the dimension of the vector space where the unknown signal is rearranged as a structured low-rank matrix is larger than the dimension of the vector space where measurements are sampled. In other words, $T$ is a tall matrix and is not surjective.
With this nontrivial difference, the conventional approaches \cite{gross2011recovering,candes2011probabilistic}
do not apply directly to structured matrix completion.
Chen and Chi \cite{chen2014robust} cleverly modified the definition of an inexact dual certificate
and the golfing scheme accordingly and provided performance guarantee at a near optimal sample complexity.
We adopt their approach to recovery of sparse signals in a transform domain,
where the nuclear norm is replaced by the $\ell_1$-norm\footnote{A subset of the authors \cite{ye2016compressive} sharpened the original analysis of the completion of structured low-rank matrices by Chen and Chi \cite{chen2014robust} particularly on the noise propagation in the recovery. In this paper, we generalize the improved version \cite{ye2016compressive}.}.
Notably, we extend the theory to the case where $T$ is not necessarily a unitary transform (e.g., $T^* T = I_n$),
which is the case in various practical applications.

We first present the following lemma that extends the notion of an inexact dual certificate
for recovery of sparse signals in a transform domain where the transform $T$ is not necessarily unitary.

\begin{lemma}[Uniqueness by an Inexact Dual Certificate]
\label{lemma:uniqueness}
Suppose that $T: \cz^n \to \cz^N$ is injective.
Let $J \subset [N]$ denote the support of $T x$,
i.e. the elements of $J$ correspond to the locations of the nonzero elements in $T x$.
Let $\Omega$ be a multi-set that consists of elements in $[n]$ with possible repetitions.
Let $\Omega'$ denote the set of all distinct elements in $\Omega$.
Suppose that
\begin{equation}
\label{eq:local_isometry_rank_deficient}
\left\|
\frac{n}{m} \Pi_J T S_\Omega^* S_\Omega T^\dagger \Pi_J
- \Pi_J T T^\dagger \Pi_J
\right\|
\leq \frac{1}{2}.
\end{equation}
If there exists a vector $v \in \cz^N$ satisfying
\begin{equation}
\label{eq:dualcert_vanish}
(T T^\dagger - T S_{\Omega'}^* S_{\Omega'} T^\dagger)^* v = 0,
\end{equation}
\begin{equation}
\label{eq:dualcert_sgn}
\norm{\Pi_J (v - \mbox{\upshape sgn}(T x))}_2 \leq \frac{1}{7 n \norm{T}_{1 \to 2} \norm{T^\dagger}_{2 \to \infty}},
\end{equation}
and
\begin{equation}
\label{eq:dualcert_bnd}
\norm{(I_N - \Pi_J) v}_\infty \leq \frac{1}{2},
\end{equation}
then $x$ is the unique minimizer to (\ref{eq:ell1min}).
\end{lemma}

\begin{proof}[Proof of Lemma~\ref{lemma:uniqueness}]
See Section~\ref{subsec:pf:lemma:uniqueness}.
\end{proof}

The next lemma shows that such an inexact dual certificate exists with hight probability
under the hypothesis of Theorem~\ref{thm:uniqueness}.

\begin{lemma}[Existence of an Inexact Dual Certificate]
\label{lemma:existence}
Suppose the hypotheses of Theorem~\ref{thm:uniqueness} hold.
Then, with probability $1 - e^{-\beta} - 1/n$,
there exists a vector $v \in \cz^N$ satisfying \cref{eq:dualcert_vanish,eq:dualcert_sgn,eq:dualcert_bnd}.
\end{lemma}

\begin{proof}[Proof of Lemma~\ref{lemma:existence}]
See Section~\ref{subsec:pf:lemma:existence}.
\end{proof}

Lemma~\ref{lemma:existence} together with Lemma~\ref{lemma:uniqueness} implies Theorem~\ref{thm:uniqueness}.
Indeed, we only need to verify \eqref{eq:local_isometry_rank_deficient}.
By Lemma~\ref{lemma:E1}, \eqref{eq:local_isometry_rank_deficient} holds with probability at least $1 - 2/n$
provided that $m \geq 32 \mu s \max\{1/(1-\norm{T^* T - I_n}), 1/6\} \log n$.
This completes the proof of Theorem~\ref{thm:uniqueness}.

In the next section, we will introduce fundamental estimates that will be used in the proofs of Lemmas~\ref{lemma:uniqueness} and \ref{lemma:existence}.

\subsection{Lemmas on fundamental estimates}
\label{subsec:riplike_lemmas}

The following lemmas provides estimates on various functions of the random matrix $T S_\Omega^* S_\Omega T^\dagger$,
which are analogous to the corresponding estimates for RIPless compressed sensing \cite{candes2011probabilistic}.

Similarly to RIPless compressed sensing, we also employ the notion of incoherence.
In fact, our incoherence assumption in \eqref{eq:incoherence} is analogous to
a generalized version for anisotropic compressed sensing \cite{rudelson2013reconstruction,lee2013oblique,kueng2014ripless}.

One important distinction from compressed sensing is that the \textit{isotropy} property is not satisfied.
Indeed, since the random indices in $\Omega$ are i.i.d. following the uniform distribution on $[n]$,
the random matrix $T S_\Omega^* S_\Omega T^\dagger$ satisfies
\[
\frac{n}{m} \mathbb{E} T S_\Omega^* S_\Omega T^\dagger = T T^\dagger.
\]
While $T T^\dagger$ is an orthogonal projection (idempotent and self-adjoint), it is not necessarily an identity operator.
The rank-deficiency of $T T^\dagger$ requires new analysis in Lemmas~\ref{lemma:uniqueness} and \ref{lemma:existence}
compared to the previous results for compressed sensing with the canonical sparsity model \cite{candes2011probabilistic}.
However, the fundamental estimates measure the deviations of functions of $T S_\Omega^* S_\Omega T^\dagger$
from their expectations and do not require the isotropy property ($T^* T = I_n$).
We present the following lemmas for the fundamental estimates, whose proofs are deferred to the appendix.

\begin{lemma}[E1: Local Isometry on a Proper Subspace]
\label{lemma:E1}
Suppose $T: \cz^n \to \cz^N$ and $\widetilde{T} = (T^\dagger)^*$ satisfy \eqref{eq:incoherence} with parameter $\mu$.
Let $\gamma$ be defined in \eqref{eq:defgamma}.
Let $J$ be a fixed subset of $[N]$ satisfying $|J| = s$.
Let $\Omega = \{\omega_1,\ldots,\omega_m\}$ where $\omega_j$s are i.i.d. following the uniform distribution on $[n]$.
Then for $\delta > 0$,
\[
\mathbb{P}\left(\left\| \Pi_J \left(\frac{n}{m} T S_\Omega^* S_\Omega T^\dagger - T T^\dagger \right) \Pi_J \right\| \geq \delta\right)
\leq 2s \exp\left( -\frac{m}{s\mu} \cdot \frac{\delta^2/2}{1/(1-\norm{\gamma T^* T - I_n})+\delta/3} \right).
\]
\end{lemma}

\begin{proof}[Proof of Lemma~\ref{lemma:E1}]
See Appendix~\ref{sec:pf:lemma:E1}.
\end{proof}

\begin{lemma}[E2: Low Distortion]
\label{lemma:E2}
Suppose $T: \cz^n \to \cz^N$ and $\widetilde{T} = (T^\dagger)^*$ satisfy \eqref{eq:incoherence} with parameter $\mu$.
Let $\gamma$ be defined in \eqref{eq:defgamma}.
Let $q \in \cz^N$ be a fixed vector.
Let $J$ be a fixed subset of $[N]$ satisfying $|J| = s$.
Let $\Omega = \{\omega_1,\ldots,\omega_m\}$ where $\omega_j$s are i.i.d. following the uniform distribution on $[n]$.
Then for each $t \leq 1/2$,
\[
\mathbb{P}\left(\left\|\Pi_J \left(\frac{n}{m} T S_\Omega^* S_\Omega T^\dagger - T T^\dagger \right) \Pi_J q\right\|_2 \geq t \norm{\Pi_J q}_2\right)
\leq \exp\left\{ -\frac{1}{4} \left( t \sqrt{\frac{m(1-\norm{\gamma T^* T - I_n})}{s\mu}} - 1 \right)^2 \right\}.
\]
\end{lemma}

\begin{proof}[Proof of Lemma~\ref{lemma:E2}]
See Appendix~\ref{sec:pf:lemma:E2}.
\end{proof}

\begin{lemma}[E3: Off-Support Incoherence]
\label{lemma:E3}
Suppose $T: \cz^n \to \cz^N$ and $\widetilde{T} = (T^\dagger)^*$ satisfy \eqref{eq:incoherence} with parameter $\mu$.
Let $\gamma$ be defined in \eqref{eq:defgamma}.
Let $q \in \cz^N$ be a fixed vector.
Let $J$ be a fixed subset of $[N]$ satisfying $|J| = s$.
Let $\Omega = \{\omega_1,\ldots,\omega_m\}$ where $\omega_j$s are i.i.d. following the uniform distribution on $[n]$.
Then for each $t > 0$,
\begin{align*}
{} & \mathbb{P}\left(\left\|\Pi_{[N] \setminus J} \left( \frac{n}{m} \widetilde{T} S_\Omega^* S_\Omega T^* - \widetilde{T} T^* \right) \Pi_J q\right\|_\infty \geq t \norm{\Pi_J q}_2\right) \\
{} & \quad \leq 2N \exp\left( - \frac{m}{2\mu} \cdot \frac{t^2}{1/(1-\norm{\gamma T^* T - I_n}) + \sqrt{s}t/3} \right).
\end{align*}
\end{lemma}

\begin{proof}[Proof of Lemma~\ref{lemma:E3}]
See Appendix~\ref{sec:pf:lemma:E3}.
\end{proof}

\begin{lemma}[E3': Off-Group-Support Incoherence]
\label{lemma:E3'}
Suppose $T: \cz^n \to \cz^L$ and $\widetilde{T} = (T^\dagger)^*$ satisfy \eqref{eq:incoherence_gs} with parameter $\mu$.
Let $\gamma$ be defined in \eqref{eq:defgamma}.
Let $q \in \cz^L$ be a fixed vector.
Let $\calG = \{\calG_1, \ldots, \calG_N\}$ be a partition of $[L]$.
Let $J$ be a fixed subset of $[N]$ satisfying $|J| = s$.
Let $\calG_J = \bigcup_{j \in J} \calG_j$.
Let $\Omega = \{\omega_1,\ldots,\omega_m\}$ where $\omega_j$s are i.i.d. following the uniform distribution on $[n]$.
Then for each $t > 0$,
\begin{align*}
{} & \mathbb{P}\left( \left\|\Pi_{[L] \setminus \calG_J} \left( \frac{n}{m} \widetilde{T} S_\Omega^* S_\Omega T^* - \widetilde{T} T^* \right) \Pi_{\calG_J} q\right\|_{\calG,\infty} \geq t \norm{\Pi_{\calG_J} q}_2\right) \\
{} & \quad \leq 2 \max_{j \in [N]} |\calG_j| N \exp\left( - \frac{m}{2\mu} \cdot \frac{t^2}{1/(1-\norm{\gamma T^* T - I_n}) + \sqrt{s}t/3} \right).
\end{align*}
\end{lemma}

\begin{proof}[Proof of Lemma~\ref{lemma:E3'}]
See Appendix~\ref{sec:pf:lemma:E3'}.
\end{proof}

\subsection{Proof of Lemma~\ref{lemma:uniqueness}}
\label{subsec:pf:lemma:uniqueness}
Our proof essentially adapts the arguments of Chen and Chi \cite[Appendix~B]{chen2014robust}
for the structured low-rank matrix completion problem.
There are two key differences in the two proofs.
First, the $\ell_1$-norm replaces the nuclear norm.
Second, $T$ is a general injective transform, which is not necessarily unitary.
These differences require nontrivial modifications of crucial steps in the proof.
Furthermore, the upper bound on the deviation of $v$ from $\mbox{sgn}(T x)$ in \eqref{eq:dualcert_sgn} is sharpened by optimizing parameters.
This improvement also applies to the previous work \cite{chen2014robust}.

Let $\hat{x} = x + h$ be the minimizer to (\ref{eq:ell1min}). We show that $T h = 0$ in two complementary cases.
Then by the injectivity of $T$, $h = 0$, or equivalently, $\hat{x} = x$.

\noindent\textbf{Case 1:} We first consider the case when $h$ satisfies
\begin{equation}
\label{eq:case1}
\norm{\Pi_J T h}_2 \leq 3 n \norm{T}_{1 \to 2} \norm{T^\dagger}_{2 \to \infty} \norm{\Pi_{[N] \setminus J} T h}_2.
\end{equation}

Since $\mbox{sgn}(\Pi_{[N] \setminus J} T h)$ and $\mbox{sgn}(T x)$ have disjoint supports,
it follows that $\mbox{sgn}(T x) + \mbox{sgn}(\Pi_{[N] \setminus J} T h)$ and $\mbox{sgn}(T x)$ coincide on $J$.
Furthermore, $\norm{\mbox{sgn}(T x) + \mbox{sgn}(\Pi_{[N] \setminus J} T h)}_\infty \leq 1$.
Therefore, $\mbox{sgn}(T x) + \mbox{sgn}(\Pi_{[N] \setminus J} T h)$ is a valid sub-gradient of the $\ell_1$-norm at $T x$.
Then it follows that
\begin{equation}
\label{eq:pf_lemma_uniqueness:ineq1}
\begin{aligned}
\norm{T x + T h}_1
{} & \geq \norm{T x}_1 + \langle \mbox{sgn}(T x) + \mbox{sgn}(\Pi_{[N] \setminus J} T h), ~ Th \rangle \\
{} & = \norm{T x}_1
+ \langle v, T h \rangle
+ \langle \mbox{sgn}(\Pi_{[N] \setminus J} T h), ~ Th \rangle
- \langle v - \mbox{sgn}(T x), Th \rangle.
\end{aligned}
\end{equation}

In fact, $\langle v, T h \rangle  = 0$ as shown below.
The inner product of $T h$ and $v$ is decomposed as
\begin{equation}
\label{eq:pf_lemma_uniqueness:ip}
\langle v, Th \rangle
= \langle v, (I_N - T T^\dagger) Th) \rangle
+ \langle v, (T T^\dagger - T S_{\Omega'}^* S_{\Omega'} T^\dagger) Th \rangle
+ \langle v, T S_{\Omega'}^* S_{\Omega'} T^\dagger Th \rangle.
\end{equation}
Indeed, all three terms in the right-hand side of \eqref{eq:pf_lemma_uniqueness:ip} are 0.
Since $T T^\dagger$ is the orthogonal projection onto the range space of $T$, the first term is 0.
The second term is 0 by the assumption on $v$ in \eqref{eq:dualcert_vanish}.
Since $\hat{x}$ is feasible for \eqref{eq:ell1min}, $S_\Omega \hat{x} = S_\Omega x$.
Thus $S_\Omega h = S_\Omega (\hat{x} - x) = 0$,
i.e. $e_\omega^* h = 0$ for all $\omega \in \Omega$, which also implies $S_{\Omega'} h = 0$.
Then it follows that $S_{\Omega'} T^\dagger T h = S_{\Omega'} h = 0$.
Thus the third term of the right-hand side of \eqref{eq:pf_lemma_uniqueness:ip} is 0.

Since the $\mbox{sgn}(\cdot)$ operator commutes with $\Pi_{[N] \setminus J}$ and $\Pi_{[N] \setminus J}$ is idempotent, we get
\begin{align*}
\langle \mbox{sgn}(\Pi_{[N] \setminus J} T h), ~ T h \rangle
{} & = \langle \Pi_{[N] \setminus J} \mbox{sgn}(\Pi_{[N] \setminus J} T h), ~ T h \rangle \\
{} & = \langle \mbox{sgn}(\Pi_{[N] \setminus J} T h), ~ \Pi_{[N] \setminus J} T h \rangle \\
{} & = \norm{\Pi_{[N] \setminus J} T h}_1.
\end{align*}

Then \eqref{eq:pf_lemma_uniqueness:ineq1} implies
\begin{equation}
\label{eq:pf_lemma_uniqueness:ineq2}
\norm{T x + T h}_1 \geq \norm{T x}_1 + \norm{\Pi_{[N] \setminus J} T h}_1 - \langle v - \mbox{sgn}(T x), Th \rangle.
\end{equation}

We derive an upper bound on the magnitude of the last term in the right-hand side of \eqref{eq:pf_lemma_uniqueness:ineq2} given by
\begin{subequations}
\label{eq:pf_lemma_uniqueness:ineq3}
\begin{align}
| \langle v - \mbox{sgn}(T x), Th \rangle |
{} & = | \langle \Pi_J (v - \mbox{sgn}(T x)), Th \rangle + \langle \Pi_{[N] \setminus J} (v - \mbox{sgn}(T x)), Th \rangle | \nonumber \\
{} & \leq | \langle \Pi_J (v - \mbox{sgn}(T x)), Th \rangle | + | \langle \Pi_{[N] \setminus J} v, Th \rangle | \label{eq:pf_lemma_uniqueness:ineq3a} \\
{} & \leq \norm{\Pi_J (v - \mbox{sgn}(T x))}_2 \norm{\Pi_J Th}_2
+ \norm{\Pi_{[N] \setminus J} v}_\infty \norm{\Pi_{[N] \setminus J} T h}_1 \label{eq:pf_lemma_uniqueness:ineq3b} \\
{} & \leq \frac{1}{7 n \norm{T}_{1 \to 2} \norm{T^\dagger}_{2 \to \infty}} \norm{\Pi_J Th}_2 + \frac{1}{2} \norm{\Pi_{[N] \setminus J} T h}_1, \label{eq:pf_lemma_uniqueness:ineq3c}
\end{align}
\end{subequations}
where \eqref{eq:pf_lemma_uniqueness:ineq3a} holds by the triangle inequality and the fact that $T x$ is supported on $J$;
\eqref{eq:pf_lemma_uniqueness:ineq3b} by H\"older's inequality;
\eqref{eq:pf_lemma_uniqueness:ineq3c} by the assumptions on $v$ in \eqref{eq:dualcert_sgn} and \eqref{eq:dualcert_bnd}.

We continue by applying \eqref{eq:pf_lemma_uniqueness:ineq3} to \eqref{eq:pf_lemma_uniqueness:ineq2} and get
\begin{align*}
\norm{T x + T h}_1
{} & \geq \norm{T x}_1 - \frac{1}{7 n \norm{T}_{1 \to 2} \norm{T^\dagger}_{2 \to \infty}} \norm{\Pi_J Th}_2 + \frac{1}{2} \norm{\Pi_{[N] \setminus J} T h}_1 \\
{} & \geq \norm{T x}_1 - \frac{3}{7} \norm{\Pi_{[N] \setminus J} T h}_2 + \frac{1}{2} \norm{\Pi_{[N] \setminus J} T h}_2 \\
{} & = \norm{T x}_1 + \frac{1}{14} \norm{\Pi_{[N] \setminus J} T h}_2,
\end{align*}
where the second step follows from \eqref{eq:case1}.

Then, $\norm{T \hat{x}}_1 \geq \norm{T x}_1 \geq \norm{T \hat{x}}_1$, which implies $\Pi_{[N] \setminus J} T h = 0$.
By (\ref{eq:case1}), we also have $\Pi_J T h = 0$.
Therefore, it follows that $T h = 0$.

\noindent\textbf{Case 2:} Next, we consider the complementary case when $h$ satisfies
\begin{equation}
\label{eq:case2}
\norm{\Pi_J T h}_2 > 3 n \norm{T}_{1 \to 2} \norm{T^\dagger}_{2 \to \infty} \norm{\Pi_{[N] \setminus J} T h}_2.
\end{equation}

In the previous case, we have shown that $S_\Omega h = 0$. Thus $S_\Omega T^\dagger T h = 0$.
Then together with $(I_N - T T^\dagger) T = 0$, we get
\begin{align*}
\left( \frac{n}{m} T S_\Omega^* S_\Omega T^\dagger + I_N - T T^\dagger \right) T h = 0,
\end{align*}
which implies
\begin{equation}
\label{eq:pf_lemma_uniqueness:ineq4}
\begin{aligned}
0 {} & \geq \left\langle \Pi_J T h, \left( \frac{n}{m} T S_\Omega^* S_\Omega T^\dagger + I_N - T T^\dagger \right) T h \right\rangle \\
{} & = \left\langle \Pi_J T h, \left( \frac{n}{m} T S_\Omega^* S_\Omega T^\dagger + I_N - T T^\dagger \right) \Pi_J T h \right\rangle \\
{} & \quad + \left\langle \Pi_J T h, \left( \frac{n}{m} T S_\Omega^* S_\Omega T^\dagger + I_N - T T^\dagger \right) \Pi_{[N] \setminus J} T h \right\rangle.
\end{aligned}
\end{equation}

The magnitude of the first term in the right-hand side of \eqref{eq:pf_lemma_uniqueness:ineq4} is lower-bounded by
\begin{equation}
\label{eq:pf_lemma_uniqueness:lb}
\begin{aligned}
{} & \left| \left\langle \Pi_J T h, \left( \frac{n}{m} T S_\Omega^* S_\Omega T^\dagger + I_N - T T^\dagger \right) \Pi_J T h \right\rangle \right| \\
{} & = \left| \langle \Pi_J T h, \Pi_J T h \rangle \right|
- \left| \left\langle \Pi_J T h, \left( T T^\dagger - \frac{n}{m} T S_\Omega^* S_\Omega T^\dagger \right) \Pi_J T h \right\rangle \right| \\
{} & \geq \norm{\Pi_J T h}_2^2 - \left\|\Pi_J T T^\dagger \Pi_J - \frac{n}{m} \Pi_J T S_\Omega^* S_\Omega T^\dagger \Pi_J\right\| \norm{\Pi_J T h}_2^2 \\
{} & \geq \frac{1}{2} \norm{\Pi_J T h}_2^2,
\end{aligned}
\end{equation}
where the last step follows from the assumption in \eqref{eq:local_isometry_rank_deficient}.

Next, we derive an upper bound on the second term in the right-hand side of \eqref{eq:pf_lemma_uniqueness:ineq4}.
To this end, we first computes the operator norm of $T e_k e_k^* T^\dagger$ for $k \in [n]$.
In fact, $\norm{T e_k e_k^* T^\dagger} = \norm{T e_k}_2 \norm{\widetilde{T} e_k}_2$, where $\widetilde{T}$ is the adjoint of $T^\dagger$.
Therefore, we only need to compute $\norm{T e_k}_2$ and $\norm{\widetilde{T} e_k}_2$.
First, $\norm{T e_k}_2$ is upper-bounded by
\begin{align*}
\max_{k \in [n]} \norm{T e_k}_2 = \norm{T}_{1 \to 2}.
\end{align*}
On the other hand, $\norm{\widetilde{T} e_k}_2$ is upper-bounded by
\begin{align*}
\max_{k \in [n]} \norm{\widetilde{T} e_k}_2
= \norm{\widetilde{T}}_{1 \to 2} = \norm{T^\dagger}_{2 \to \infty},
\end{align*}
where the last step holds since $\widetilde{T}$ is the adjoint operator of $T^\dagger$
and $\ell_\infty^n$ is the dual space of $\ell_1^n$.
By the above upper bounds on $\norm{T e_k}_2$ and $\norm{\widetilde{T} e_k}_2$, we get
\begin{align*}
\norm{T e_k e_k^* T^\dagger} \leq \norm{T}_{1 \to 2} \norm{T^\dagger}_{2 \to \infty}.
\end{align*}
Then the operator norm of $ \frac{n}{m} T S_\Omega^* S_\Omega T^\dagger + I_N - T T^\dagger $ is upper-bounded by
\begin{equation}
\label{eq:pf_lemma_uniqueness:ineq5}
\begin{aligned}
\left\|  \frac{n}{m} T S_\Omega^* S_\Omega T^\dagger + I_N - T T^\dagger  \right\|
{} & \leq \frac{n}{m} \left( \norm{T e_{\omega_1} e_{\omega_1}^* T^\dagger + I_N - T T^\dagger} + \sum_{j=2}^m \norm{T e_{\omega_j} e_{\omega_j}^* T^\dagger}_2 \right) \\
{} & \leq \frac{n}{m} \left( \max(\norm{T e_{\omega_1} e_{\omega_1}^* T^\dagger}, \norm{I_N - T T^\dagger}) + \sum_{j=2}^m \norm{T e_{\omega_j} e_{\omega_j}^* T^\dagger}_2 \right) \\
{} & \leq n \norm{T}_{1 \to 2} \norm{T^\dagger}_{2 \to \infty},
\end{aligned}
\end{equation}
where the second step follows since
$(T e_{\omega_1} e_{\omega_1}^* T^\dagger)^*(I_N - T T^\dagger) = 0$ and
$(I_N - T T^\dagger) (T e_{\omega_1} e_{\omega_1}^* T^\dagger)^* = 0$.
The second term in the right-hand side of \eqref{eq:pf_lemma_uniqueness:ineq4} is then upper-bounded by
\begin{equation}
\label{eq:pf_lemma_uniqueness:ub}
\begin{aligned}
{} & \left| \left\langle \Pi_J T h, \left( \frac{n}{m} T S_\Omega^* S_\Omega T^\dagger + I_N - T T^\dagger \right) \Pi_{[N] \setminus J} T h \right\rangle \right| \\
{} & \leq \left\|  \frac{n}{m} T S_\Omega^* S_\Omega T^\dagger + I_N - T T^\dagger  \right\| \norm{\Pi_J T h}_2 \norm{\Pi_{[N] \setminus J} T h}_2 \\
{} & \leq n \norm{T}_{1 \to 2} \norm{T^\dagger}_{2 \to \infty} \norm{\Pi_J T h}_2 \norm{\Pi_{[N] \setminus J} T h}_2,
\end{aligned}
\end{equation}
where the last step follows from \eqref{eq:pf_lemma_uniqueness:ineq5}.

Applying \eqref{eq:pf_lemma_uniqueness:lb} and \eqref{eq:pf_lemma_uniqueness:ub} to \eqref{eq:pf_lemma_uniqueness:ineq4} provides
\begin{align*}
0 {} & \geq \left| \left\langle \Pi_J T h, \left( \frac{n}{m} T S_\Omega^* S_\Omega T^\dagger + I_N - T T^\dagger \right) \Pi_J T h \right\rangle \right| \\
{} & \quad - \left| \left\langle \Pi_J T h, \left( \frac{n}{m} T S_\Omega^* S_\Omega T^\dagger + I_N - T T^\dagger \right) \Pi_{[N] \setminus J} T h \right\rangle \right| \\
{} & \geq \frac{1}{2} \norm{\Pi_J T h}_2^2 - n \norm{T}_{1 \to 2} \norm{T^\dagger}_{2 \to \infty} \norm{\Pi_J T h}_2 \norm{\Pi_{[N] \setminus J} T h}_2 \\
{} & \geq \frac{1}{2} \norm{\Pi_J T h}_2^2 - \frac{1}{3} \norm{\Pi_J T h}_2^2 \\
{} & = \frac{1}{6} \norm{\Pi_J T h}_2^2 \geq 0,
\end{align*}
where the second inequality follows from \eqref{eq:case2}.
Then it is implied that $\Pi_J T h = 0$. By (\ref{eq:case2}), we also have $\Pi_{[N] \setminus J} T h = 0$.
Therefore, $T h = 0$, which completes the proof.

\subsection{Proof of Lemma~\ref{lemma:existence}}
\label{subsec:pf:lemma:existence}
We construct a dual certificate $v$ using a golfing scheme.
Since the isotropy is not satisfied, the original golfing scheme needs to be modified accordingly.
We adopt the version for structured matrix completion \cite{chen2014robust}.

Recall that the elements of $\Omega = \{\omega_1, \ldots, \omega_m\}$ are i.i.d. following the uniform distribution on $[n]$.
We partition the multi-set $\Omega$ into $\ell$ multi-sets so that
$\Omega_1$ consists of the first $m_1$ elements of $\Omega$,
$\Omega_2$ consists of the next $m_2$ elements of $\Omega$, and so on, where $\sum_{i=1}^\ell m_i = m$.
Then, $\Omega_i$s are mutually independent and each $\Omega_i$ consists of i.i.d. random indices.

The version of the golfing scheme in this paper generates a dual certificate $v \in \cz^N$
from intermediate vectors $q_i \in \cz^N$ for $i=0,\ldots,\ell-1$ by
\[
v = \sum_{i=1}^\ell \left(\frac{n}{m_i} \widetilde{T} S_{\Omega_i}^* S_{\Omega_i} T^* + I_N - T T^\dagger \right) q_{i-1},
\]
where $q_i$s are generated as follows:
first, initialize $q_0 = \mbox{sgn}(T x)$;
next, generate $q_i$s recursively by
\[
q_i = \Pi_J \left( \widetilde{T} T^* - \frac{n}{m_i} \widetilde{T} S_{\Omega_i}^* S_{\Omega_i} T^* \right) q_{i-1}, \quad i = 1,\ldots,\ell-1.
\]
Here, $\widetilde{T}$ denotes the adjoint of $T^\dagger$.

Note that
\begin{align*}
{} & (T T^\dagger - T S_{\Omega'}^* S_{\Omega'} T^\dagger)^*
\left( \frac{n}{m_i} \widetilde{T} S_{\Omega_i}^* S_{\Omega_i} T^* + I_N - \widetilde{T} T^* \right) \\
{} & = \frac{n}{m_i} (T T^\dagger - T S_{\Omega'}^* S_{\Omega'} T^\dagger)^* \widetilde{T} S_{\Omega_i}^* S_{\Omega_i} T^*
+ (T T^\dagger - T S_{\Omega'}^* S_{\Omega'} T^\dagger)^* (I_N - \widetilde{T} T^*) \\
{} & = \frac{n}{m_i} (\widetilde{T} T^* - \widetilde{T} S_{\Omega'}^* S_{\Omega'} T^*) \widetilde{T} S_{\Omega_i}^* S_{\Omega_i} T^*
+ (\widetilde{T} T^* - \widetilde{T} S_{\Omega'}^* S_{\Omega'} T^*) (I_N - \widetilde{T} T^*) = 0.
\end{align*}
Thus it follows that $v$ satisfies \eqref{eq:dualcert_vanish}.

The rest of the proof is devoted to show that $v$ satisfies \eqref{eq:dualcert_sgn} and \eqref{eq:dualcert_bnd},
which follows similarly to the proof of \cite[Lemma~3.3]{candes2011probabilistic}.
For completeness, we verify that the arguments in \cite{candes2011probabilistic} are valid in our setting
(with neither isotropy nor self-adjointness).

We show that $q_i$ satisfies the following two properties with high probability for each $i \in [\ell]$:
first,
\begin{equation}
\label{eq:decay_qi}
\norm{q_i}_2 \leq c_i \norm{q_{i-1}}_2
\end{equation}
and, second,
\begin{equation}
\label{eq:bnd_qi_infty}
\left\| \Pi_{[N] \setminus J} \left(\frac{n}{m_i} \widetilde{T} S_{\Omega_i}^* S_{\Omega_i} T^* + I_N - \widetilde{T} T^* \right) q_{i-1} \right\|_\infty \leq t_i \norm{q_{i-1}}_2.
\end{equation}

Let $p_1(i)$ (resp. $p_2(i)$) denote the probability that the inequality in \eqref{eq:decay_qi} (resp. \eqref{eq:bnd_qi_infty}) does not hold.
Since $q_{i-1}$ is independent of $\Omega_i$, by Lemma~\ref{lemma:E2}, $p_1(i)$ is upper-bounded by
\[
p_1(i) \leq \exp\left( - \frac{1}{4}(c_i \sqrt{m_i(1-\norm{\gamma T^* T - I_n})/(s\mu)}-1)^2 \right).
\]
Therefore, $p_1(i) \leq \frac{1}{\alpha} e^{-\beta}$ if
\begin{equation}
\label{eq:bndmi_p1}
m_i \geq \frac{2+8(\beta + \log \alpha)}{c_i^2} \cdot \frac{\mu s}{1-\norm{\gamma T^* T - I_n}}.
\end{equation}
On the other hand, note that $q_{i-1} = \Pi_J q_{i-1}$. Then it follows that
\begin{align*}
\left\| \Pi_{[N] \setminus J} \left(\frac{n}{m_i} \widetilde{T} S_{\Omega_i}^* S_{\Omega_i} T^* + I_N - \widetilde{T} T^* \right) q_{i-1} \right\|_\infty
{} & = \left\| \Pi_{[N] \setminus J} \left(\frac{n}{m_i} \widetilde{T} S_{\Omega_i}^* S_{\Omega_i} T^* + I_N - \widetilde{T} T^* \right) \Pi_J q_{i-1} \right\|_\infty \\
{} & = \left\| \Pi_{[N] \setminus J} \left(\frac{n}{m_i} \widetilde{T} S_{\Omega_i}^* S_{\Omega_i} T^* - \widetilde{T} T^* \right) \Pi_J q_{i-1} \right\|_\infty.
\end{align*}
Again, we use the fact that $q_{i-1}$ is independent of $\Omega_i$.
Then, by Lemma~\ref{lemma:E3}, $p_2(i)$ is upper-bounded by
\[
p_2(i) \leq 2N \exp\left( - \frac{3t_i^2 m_i}{6\mu/(1-\norm{\gamma T^* T - I_n}) + 2 \mu \sqrt{s} t_i} \right).
\]
Therefore, $p_2(i) \leq \frac{1}{\alpha} e^{-\beta}$ if
\begin{equation}
\label{eq:bndmi_p2}
m_i \geq \left( \frac{2}{t_i^2 s (1-\norm{\gamma T^* T - I_n})} + \frac{2}{3t_i\sqrt{s}} \right) (\beta + \log(2\alpha) + \log N) s\mu.
\end{equation}

We set the parameters similarly to the proof of \cite[Lemma~3.3]{candes2011probabilistic} as follows:
\begin{equation}
\label{eq:param}
\begin{aligned}
\ell {} & = \left\lceil \frac{\log_2 s}{2} + \log_2 n
+ \log_2 \left( \norm{T}_{1 \to 2} \norm{T^\dagger}_{2 \to \infty} \right) \right\rceil + 3, \\
c_i {} & =
\begin{cases}
1/\lceil 2 \sqrt{\log N} ~ \rceil & i=1,2,3, \\
1/2 & 3 \leq i \leq \ell,
\end{cases} \\
t_i {} & =
\begin{cases}
1/\lceil 4 \sqrt{s} ~\rceil & i=1,2,3, \\
\log N / \lceil 4 \sqrt{s} ~\rceil & 3 \leq i \leq \ell,
\end{cases}\\
m_i {} & = \lceil 10(1+\log 6+\beta) \mu s c_i^{-2} \rceil, \quad \forall i. \\
\end{aligned}
\end{equation}
By the construction of $v$, we have
\begin{align*}
\Pi_J v {} & = \sum_{i=1}^\ell \Pi_J \left(\frac{n}{m_i} \widetilde{T} S_{\Omega_i}^* S_{\Omega_i} T^* + I_N - \widetilde{T} T^* \right) q_{i-1} \\
{} & = \sum_{i=1}^\ell \left[ \Pi_J q_{i-1} - \Pi_J \left(\widetilde{T} T^* - \frac{n}{m_i} \widetilde{T} S_{\Omega_i}^* S_{\Omega_i} T^* \right) q_{i-1} \right] \\
{} & = \sum_{i=1}^\ell \left( q_{i-1} - q_i \right)
= q_0 - q_\ell \\
{} & = \mbox{sgn}(T x) - q_\ell = \Pi_J \mbox{sgn}(T x) - q_\ell.
\end{align*}
Therefore, \eqref{eq:decay_qi} implies
\[
\norm{\Pi_J (v - \mbox{sgn}(T x))}_2
= \norm{q_\ell}_2
\leq \prod_{i=1}^\ell c_i \norm{\mbox{sgn}(T x)}_2
\leq \frac{\sqrt{s}}{2^\ell \log N}.
\]




Next, by \eqref{eq:decay_qi} and \eqref{eq:bnd_qi_infty}, we have
\begin{equation}
\label{eq:dualcert_bnd_proof}
\begin{aligned}
\norm{\Pi_{[N] \setminus J} v}_\infty
{} & \leq \sum_{i=1}^\ell \left\| \Pi_{[N] \setminus J} \left(\frac{n}{m_i} \widetilde{T} S_{\Omega_i}^* S_{\Omega_i} T^* + I_N - \widetilde{T} T^* \right) q_{i-1} \right\|_\infty \\
{} & \leq \sum_{i=1}^\ell t_i \norm{q_{i-1}}_2 \\
{} & \leq \sqrt{s} \left(t_1 + \sum_{i=2}^\ell t_i \prod_{j=1}^{i-1} c_j\right).
\end{aligned}
\end{equation}
By setting parameters as in \eqref{eq:param}, the right-hand side in \eqref{eq:dualcert_bnd_proof} is further upper-bounded by
\[
\frac{1}{4} \left( 1 + \frac{1}{2\sqrt{\log N}} + \frac{\log N}{4 \log N} + \cdots \right) < \frac{1}{2}.
\]
Then, we have shown that $v$ satisfies \eqref{eq:dualcert_bnd}.

It remains to show that \eqref{eq:decay_qi} and \eqref{eq:bnd_qi_infty} hold with the desired probability.
From \eqref{eq:bndmi_p1} and \eqref{eq:bndmi_p2}, it follows that
\[
p_j(i) \leq \frac{1}{6} e^{-\beta}, \quad \forall i \in [\ell], ~ \forall j=1,2.
\]
In particular, we have
\[
\sum_{j=1}^2 \sum_{i=1}^3 p_j(i) \leq e^{-\beta}.
\]
This implies that the first three $\Omega_i$s satisfy \eqref{eq:decay_qi} and \eqref{eq:bnd_qi_infty} except with probability $e^{-\beta}$.

On the other hand, we also have
\[
p_1(i) + p_2(i) < \frac{1}{3}, \quad \forall i = 4,\ldots,\ell.
\]
In other words, the probability that $\Omega_i$ satisfies \eqref{eq:decay_qi} and \eqref{eq:bnd_qi_infty} is at least $2/3$.
The union bound doesn't show that $\Omega_i$ satisfies \eqref{eq:decay_qi} and \eqref{eq:bnd_qi_infty} for all $i \geq 4$ with the desired probability.

As in the proof of \cite[Lemma~3.3]{candes2011probabilistic},
we adopt the oversampling and refinement strategy by Gross \cite{gross2011recovering}.
Recall that each random index set $\Omega_i$ consists of i.i.d. random indices following the uniform distribution on $[n]$.
Thus $\Omega_i$s are mutually independent.
In particular, we set $\Omega_i$s are of the same cardinality in \eqref{eq:param}.
Therefore, $\Omega_i$s are i.i.d. random variables.
We generate a few extra copies of $\Omega_i$ for $i = \ell+1,\ldots,\ell'+3$ where $\ell' = 3(\ell-3)$.
Then, by Hoeffding's inequality, there exist at least $\ell-3$ $\Omega_i$s for $i \geq 4$
that satisfy \eqref{eq:decay_qi} and \eqref{eq:bnd_qi_infty} with probability $1 - 1/n$.
(We refer more technical details for this step to \cite[Section~III.B]{candes2011probabilistic}.)
Therefore, there are $\ell$ good $\Omega_i$s satisfying \eqref{eq:decay_qi} and \eqref{eq:bnd_qi_infty}
with probability $1 - e^{-\beta} - 1/n$, and the dual certificate $v$ is constructed from these good $\Omega_i$s.
The total number of samples for this construction requires
\[
m \geq \frac{40(1+\log 6 +\beta)\mu s (3\log N + 3\ell)}{1 - \norm{\gamma T^* T - I_n}},
\]
which can be simplified as
\[
m \geq \frac{C(1+\beta) \mu s}{1 - \norm{\gamma T^* T - I_n}} \left[ \log N
+ \log \left( \norm{T}_{1 \to 2} \norm{T^\dagger}_{2 \to \infty} \right) \right]
\]
for a numerical constant $C$.

\section{Proofs for Theorems~\ref{thm:stability1} and \ref{thm:stability2}}
\label{sec:pf_main_result_noisy}

In this section, we prove Theorems~\ref{thm:stability1} and \ref{thm:stability2},
which provide sufficient conditions for stable recovery of sparse signals in a transform domain from noisy data.

\subsection{Proof of Theorem~\ref{thm:stability1}}
\label{subsec:pf:thm:stability1}
Let $h := \hat{x} - x$. Since $\hat{x}$ is the minimizer to \eqref{eq:ell1min_noisy}, it follows that
\begin{equation}
\label{eq:pf_thm_stability:bnd1}
\begin{aligned}
\norm{T x}_1
{} & \geq \norm{T \hat{x}}_1 = \norm{T x + T h}_1 \\
{} & \geq \norm{T x + T (I_n - S_{\Omega'}^* S_{\Omega'}) h}_1 - \norm{T S_{\Omega'}^* S_{\Omega'} h}_1.
\end{aligned}
\end{equation}

Since $x$ and $\hat{x}$ are feasible for \eqref{eq:ell1min_noisy}, it follows that
\begin{equation}
\label{eq:noisy_feas}
\begin{aligned}
{} & \max\left( \norm{T S_{\Omega'}^* S_{\Omega'} (\hat{x} - x^\sharp)}_2, ~ \norm{T S_{\Omega'}^* S_{\Omega'} (x - x^\sharp)}_2 \right) \\
{} & \leq \norm{T} \max\left( \norm{S_{\Omega'} (\hat{x} - x^\sharp)}_2, ~ \norm{S_{\Omega'} (x - x^\sharp)}_2 \right) \leq \epsilon \norm{T}.
\end{aligned}
\end{equation}
Therefore, by the triangle inequality, we have
\begin{equation}
\label{eq:tube_cstr}
\norm{T S_{\Omega'}^* S_{\Omega'} h}_2 = \norm{T S_{\Omega'}^* S_{\Omega'} (\hat{x} - x)}_2
\leq \norm{T S_{\Omega'}^* S_{\Omega'} (\hat{x} - x^\sharp)}_2 + \norm{T S_{\Omega'}^* S_{\Omega'} (x - x^\sharp)}_2 \leq 2 \epsilon \norm{T}.
\end{equation}

The rest of the proof will compute upper bounds on $\norm{T h}_2$ in two complementary cases similarly to the proof of Lemma~\ref{lemma:uniqueness}.
Unlike the noiseless case ($x^\sharp = x$ and $\epsilon = 0$) in Lemma~\ref{eq:bnd_qi_infty},
the condition in \eqref{eq:tube_cstr} does not necessarily imply $S_\Omega h = 0$.
In fact, the proof of Lemma~\ref{lemma:uniqueness} critically depends on the condition $S_\Omega h = 0$.
Essentially, we replace $h$ by $(I_n - S_{\Omega'}^* S_{\Omega'}) h$. Then it follows that
\begin{equation}
\label{eq:noisy_vanish}
S_{\Omega'} (I_n - S_{\Omega'}^* S_{\Omega'}) h = 0.
\end{equation}

\noindent\textbf{Case 1:} We first consider the case when $(I_n - S_{\Omega'}^* S_{\Omega'}) h$ satisfies
\begin{equation}
\label{eq:case1noisy}
\norm{\Pi_J T (I_n - S_{\Omega'}^* S_{\Omega'}) h}_2
\leq 3 n \norm{T}_{1 \to 2} \norm{T^\dagger}_{2 \to \infty} \norm{\Pi_{[N] \setminus J} T (I_n - S_{\Omega'}^* S_{\Omega'}) h}_2,
\end{equation}
where $J$ denotes the support of $T x$.

Under \eqref{eq:case1noisy}, similarly to the proof of Lemma~\ref{lemma:uniqueness}, we have
\begin{equation}
\label{eq:pf_thm_stability:bnd2}
\norm{T x + T (I_n - S_{\Omega'}^* S_{\Omega'}) h}_1
\geq \norm{T x}_1 + \frac{1}{14} \norm{\Pi_{[N] \setminus J} T (I_n - S_{\Omega'}^* S_{\Omega'}) h}_2.
\end{equation}
Combining \eqref{eq:pf_thm_stability:bnd1} and \eqref{eq:pf_thm_stability:bnd2} provides
\begin{equation}
\label{eq:pf_thm_stability:bnd3}
\norm{\Pi_{[N] \setminus J} T (I_n - S_{\Omega'}^* S_{\Omega'}) h}_2
\leq 14 \norm{T S_{\Omega'}^* S_{\Omega'} h}_1
\leq 14 \sqrt{N} \norm{T S_{\Omega'}^* S_{\Omega'} h}_2 \leq 28\sqrt{N} \epsilon \norm{T}.
\end{equation}
On the other hand, \eqref{eq:case1noisy} implies
\begin{equation}
\label{eq:pf_thm_stability:bnd4}
\begin{aligned}
\norm{T (I_n - S_{\Omega'}^* S_{\Omega'}) h}_2
{} & \leq \norm{\Pi_{[N] \setminus J} T (I_n - S_{\Omega'}^* S_{\Omega'}) h}_2 + \norm{\Pi_J T (I_n - S_{\Omega'}^* S_{\Omega'}) h}_2 \\
{} & \leq (1 + 3 n \norm{T}_{1 \to 2} \norm{T^\dagger}_{2 \to \infty}) \norm{\Pi_{[N] \setminus J} T (I_n - S_{\Omega'}^* S_{\Omega'}) h}_2.
\end{aligned}
\end{equation}
Therefore, combining \cref{eq:tube_cstr,eq:pf_thm_stability:bnd3,eq:pf_thm_stability:bnd4} provides
\begin{equation}
\label{eq:pf_thm_stability:bnd5}
\norm{T h}_2 \leq \left\{ 2 + 28 \sqrt{N} \left( 3 n \norm{T}_{1 \to 2} \norm{T^\dagger}_{2 \to \infty} + 1 \right) \right\} \epsilon \norm{T}.
\end{equation}

\noindent\textbf{Case 2:} Next, we consider the complementary case when $(I_n - S_{\Omega'}^* S_{\Omega'}) h$ satisfies
\begin{equation}
\label{eq:case2noisy}
\norm{\Pi_J T (I_n - S_{\Omega'}^* S_{\Omega'}) h}_2
> 3 n \norm{T}_{1 \to 2} \norm{T^\dagger}_{2 \to \infty} \norm{\Pi_{[N] \setminus J} T (I_n - S_{\Omega'}^* S_{\Omega'}) h}_2.
\end{equation}

Again, similarly to the proof of Lemma~\ref{lemma:uniqueness}, we get $(I_n - S_{\Omega'}^* S_{\Omega'}) h = 0$.
Therefore,
\[
\norm{T h}_2 \leq \norm{T (I_n - S_{\Omega'}^* S_{\Omega'}) h}_2 + \norm{T S_{\Omega'}^* S_{\Omega'} T h}_2 \leq 2 \epsilon \norm{T},
\]
which is smaller than the upper bound on $\norm{T h}_2$ in the previous case.

By applying $\norm{h}_2 \leq \norm{T h}_2 / \sigma_{\min}(T)$ to \eqref{eq:pf_thm_stability:bnd5},
we obtain the desired upper bound on $\norm{h}_2$.
This completes the proof.

\subsection{Proof of Theorem~\ref{thm:stability2}}
\label{subsec:pf:thm:stability2}
By Theorem~\ref{thm:rboplike}, the condition in \eqref{eq:noisy_samp_comp} implies that
$\frac{n}{m} T S_\Omega^* S_\Omega T^\dagger$ satisfies the condition in \eqref{eq:rboplike} with constant $\delta = 1/3$.
Note that the estimates in Lemmas~\cref{lemma:E1,lemma:E2,lemma:E3} are implied by \eqref{eq:rboplike}.
Therefore, the rest of the proof will be identical to that of Theorem~\ref{thm:stability2} except that
we compute a tighter upper bound on $\norm{\Pi_J T (I_n - S_{\Omega'}^* S_{\Omega'}) h}_2$ as follows.

First, we decompose $\Pi_J T (I_n - S_{\Omega'}^* S_{\Omega'}) h$ as
\begin{equation}
\label{eq:pf_thm_stability2:decomp}
\begin{aligned}
\Pi_J T (I_n - S_{\Omega'}^* S_{\Omega'}) h
{} & = \Pi_J \left(T T^\dagger - \frac{n}{m} T S_\Omega^* S_\Omega T^\dagger\right) T h \\
{} & \quad + \Pi_J \frac{n}{m} T (S_\Omega^* S_\Omega - S_{\Omega'}^* S_{\Omega'}) h \\
{} & \quad + \Pi_J \left( \frac{n}{m} - 1 \right) T S_{\Omega'}^* S_{\Omega'} h.
\end{aligned}
\end{equation}

Let $J_1$ correspond to the indices of the $s$-largest coefficients of $\Pi_{[N] \setminus J} T h$;
$J_2$ to the indices of the next $s$-largest coefficients of $\Pi_{[N] \setminus J} T h$, and so on.
Then the $\ell_2$-norm first term in the right-hand side of \eqref{eq:pf_thm_stability2:decomp} is upper-bounded by
\begin{equation}
\label{eq:pf_thm_stability2:bnd1}
\begin{aligned}
{} & \left\|\Pi_J \left(T T^\dagger - \frac{n}{m} T S_\Omega^* S_\Omega T^\dagger\right) T h\right\|_2 \\
{} & \leq \left\|\Pi_J \left(T T^\dagger - \frac{n}{m} T S_\Omega^* S_\Omega T^\dagger\right) \Pi_{J \cup J_1} T h\right\|_2 \\
{} & \quad + \sum_{i \geq 2} \left\|\Pi_J \left(T T^\dagger - \frac{n}{m} T S_\Omega^* S_\Omega T^\dagger\right) \Pi_{J_i} T h\right\|_2.
\end{aligned}
\end{equation}
By \eqref{eq:rboplike}, the first term in the right-hand side of \eqref{eq:pf_thm_stability2:bnd1} is upper-bounded by
\begin{equation}
\label{eq:pf_thm_stability2:bnd2}
\left\|\Pi_J \left(T T^\dagger - \frac{n}{m} T S_\Omega^* S_\Omega T^\dagger\right) \Pi_{J \cup J_1} T h\right\|_2
\leq \frac{1}{3} \norm{\Pi_{J \cup J_1} T h}_2
\leq \frac{1}{3} \norm{T h}_2.
\end{equation}

By \eqref{eq:rboplike}, the second term in the right-hand side of \eqref{eq:pf_thm_stability2:bnd1} is upper-bounded by
\begin{equation}
\label{eq:pf_thm_stability2:bnd3}
\begin{aligned}
{} & \sum_{i \geq 2} \left\|\Pi_J \left(T T^\dagger - \frac{n}{m} T S_\Omega^* S_\Omega T^\dagger\right) \Pi_{J_i} T h\right\|_2 \\
{} & \leq \frac{1}{3} \sum_{i \geq 2} \norm{\Pi_{J_i} T h}_2
\leq \frac{1}{3\sqrt{s}} \sum_{i \geq 1} \norm{\Pi_{J_i} T h}_1
= \frac{1}{3\sqrt{s}} \norm{\Pi_{[N] \setminus J} T h}_1.
\end{aligned}
\end{equation}
Since $\hat{x}$ is the minimizer to \eqref{eq:ell1min_noisy}, we have the so called ``cone'' constraint:
\[
\norm{T x}_1 \geq \norm{T \hat{x}}_1 \geq \norm{Tx + T h}_1 \geq \norm{\Pi_J T x}_1 - \norm{\Pi_J T h}_1 + \norm{\Pi_{[N] \setminus J} T h}_1 - \norm{\Pi_{[N] \setminus J} T x}_1,
\]
which implies
\[
\norm{\Pi_{[N] \setminus J} T h}_1 \leq 2 \norm{\Pi_{[N] \setminus J} T x}_1 + \norm{\Pi_J T h}_1 = \norm{\Pi_J T h}_1,
\]
where the last step follows since $T x$ is supported on $J$.
Therefore, \eqref{eq:pf_thm_stability2:bnd3} implies
\begin{equation}
\label{eq:pf_thm_stability2:bnd4}
\sum_{i \geq 2} \left\|\Pi_J \left(T T^\dagger - \frac{n}{m} T S_\Omega^* S_\Omega T^\dagger\right) \Pi_{J_i} T h\right\|_2
\leq \frac{1}{3\sqrt{s}} \norm{\Pi_{[N] \setminus J} T h}_1
\leq \frac{1}{3\sqrt{s}} \norm{\Pi_J T h}_1
\leq \frac{1}{3} \norm{\Pi_J T h}_2.
\end{equation}

Plugging \eqref{eq:pf_thm_stability2:bnd2} and \eqref{eq:pf_thm_stability2:bnd4} to \eqref{eq:pf_thm_stability2:bnd1} provides
\begin{equation}
\label{eq:pf_thm_stability2:bnd5}
\left\|\Pi_J \left(T T^\dagger - \frac{n}{m} T S_\Omega^* S_\Omega T^\dagger\right) T h\right\|_2
\leq \frac{2}{3} \norm{T h}_2.
\end{equation}

Next, we derive an upper bound on the $\ell_2$-norm of the second term in the right-hand side of \eqref{eq:pf_thm_stability2:decomp}.
Since $\Omega'$ consists of distinct elements in $\Omega$, it follows that
\[
S_{\Omega'}^* S_{\Omega'} S_\Omega^* S_\Omega
= \sum_{k \in \Omega} S_{\Omega'}^* S_{\Omega'} e_k e_k^*
= \sum_{k \in \Omega} e_k e_k^*
= S_\Omega^* S_\Omega.
\]
Furthermore, $S_{\Omega'}^* S_{\Omega'}$ is idempotent.
Therefore, we obtain the following identity:
\begin{equation}
\label{eq:pf_thm_stability2:id1}
S_\Omega^* S_\Omega - S_{\Omega'}^* S_{\Omega'} = (S_\Omega^* S_\Omega - S_{\Omega'}^* S_{\Omega'}) S_{\Omega'}^* S_{\Omega'}.
\end{equation}
From \eqref{eq:pf_thm_stability2:id1}, we get
\begin{equation}
\label{eq:pf_thm_stability2:bnd6}
\begin{aligned}
\norm{T (S_\Omega^* S_\Omega - S_{\Omega'}^* S_{\Omega'}) h}_2
{} & = \norm{T (S_\Omega^* S_\Omega - S_{\Omega'}^* S_{\Omega'}) S_{\Omega'}^* S_{\Omega'} h}_2 \\
{} & = \norm{T (S_\Omega^* S_\Omega - S_{\Omega'}^* S_{\Omega'}) T^\dagger T S_{\Omega'}^* S_{\Omega'} h}_2 \\
{} & \leq \norm{T (S_\Omega^* S_\Omega - S_{\Omega'}^* S_{\Omega'}) T^\dagger} \norm{T S_{\Omega'}^* S_{\Omega'} h}_2 \\
{} & \leq \max_k \norm{T e_k e_k^* T^\dagger} (|\Omega|-|\Omega'|) \norm{T S_{\Omega'}^* S_{\Omega'} h}_2 \\
{} & \leq \norm{T}_{1 \to 2} \norm{T^\dagger}_{2 \to \infty} (|\Omega|-|\Omega'|) \norm{T S_{\Omega'}^* S_{\Omega'} h}_2.
\end{aligned}
\end{equation}
Then the $\ell_2$-norm of the second term in the right-hand side of \eqref{eq:pf_thm_stability2:decomp} is upper-bounded by
\begin{equation}
\label{eq:pf_thm_stability2:bnd7}
\begin{aligned}
\left\|\Pi_J \frac{n}{m} T (S_\Omega^* S_\Omega - S_{\Omega'}^* S_{\Omega'}) h\right\|_2
{} & \leq \frac{n}{m} \norm{T (S_\Omega^* S_\Omega - S_{\Omega'}^* S_{\Omega'}) h}_2 \\
{} & \leq \frac{n}{m} \norm{T}_{1 \to 2} \norm{T^\dagger}_{2 \to \infty} (|\Omega|-|\Omega'|) \norm{T S_{\Omega'}^* S_{\Omega'} h}_2.
\end{aligned}
\end{equation}

By applying \eqref{eq:pf_thm_stability2:bnd5} and \eqref{eq:pf_thm_stability2:bnd7} to \eqref{eq:pf_thm_stability2:decomp},
then combining the result with \eqref{eq:tube_cstr} and \eqref{eq:pf_thm_stability:bnd3}, we get
\begin{align*}
\norm{T h}_2
\leq 2 \epsilon \norm{T} + 28 \sqrt{N} \epsilon \norm{T} + \frac{2}{3} \norm{T h}_2
+ 2 \frac{n}{m} \norm{T}_{1 \to 2} \norm{T^\dagger}_{2 \to \infty} (|\Omega|-|\Omega'|) \epsilon \norm{T}
+ 2 \left( \frac{n}{m} - 1 \right) \epsilon \norm{T},
\end{align*}
which implies
\[
\norm{T h}_2 \leq 6 \left[ 14\sqrt{N} + \frac{n}{m} \left( \norm{T}_{1 \to 2} \norm{T^\dagger}_{2 \to \infty} (|\Omega|-|\Omega'|) + 1 \right) \right] \epsilon \norm{T}.
\]

In the case when $T^* T = I_n$,
$T e_k e_k^* T^\dagger$s correspond to orthogonal projections onto mutually orthogonal one-dimensional subspaces.
Recall that the summands in $T S_\Omega^* S_\Omega T^\dagger$ repeat at most $R$ times.
Then the summands in $T S_\Omega^* S_\Omega T^\dagger - T S_{\Omega'}^* S_{\Omega'} T^\dagger$ repeat at most $R - 1$ times.
Therefore, we get a sharper estimate than that in \eqref{eq:pf_thm_stability2:bnd6} given by
\begin{equation}
\label{eq:pf_thm_stability2:bnd8}
\begin{aligned}
\norm{T (S_\Omega^* S_\Omega - S_{\Omega'}^* S_{\Omega'}) h}_2
{} & \leq \norm{T (S_\Omega^* S_\Omega - S_{\Omega'}^* S_{\Omega'}) T^\dagger} \norm{T S_{\Omega'}^* S_{\Omega'} h}_2 \\
{} & \leq (R - 1) \norm{T S_{\Omega'}^* S_{\Omega'} h}_2.
\end{aligned}
\end{equation}
This completes the proof for the second claim.


\section{Numerical Results}
\label{sec:numres}
In this section, we conduct numerical experiments of solving \eqref{eq:ell1min} and \eqref{eq:mixedmin}, with partial Fourier measurements and several different sparsifying transforms. We compare different sampling schemes in Monte-Carlo experiments, and observe that the variable density sampling schemes proposed in \eqref{eq:variable_density} and \eqref{eq:variable_density_gs} yield superior recovery results in terms of success rate.

The optimization problems \eqref{eq:ell1min} and \eqref{eq:mixedmin} are solved using Alternating Direction Method of Multipliers (ADMM) \cite{boyd2011distributed}. For example, for $T=\Phi\Psi^*\in\cz^{n\times n}$, \eqref{eq:ell1min} is rewritten in the following form with two linear constraints, and solved by ADMM with two additional terms in the augmented Lagrangian.
\[
\begin{array}{ll}
\displaystyle \minimize_{y \in \cz^n} & \norm{y}_1 \\
\mathrm{subject~to} & T g = y,\\
                    & S_\Omega g = S_\Omega x.
\end{array}
\]
In all experiments, the ADMM algorithm runs for 1,000 iterations.

\subsection{1D Signals}
For 1D signals, we use the DFT $\Psi\in\cz^{n\times n}$, and two sparsifying transforms $\Psi_{\mathrm{TV}, n}, \Psi_{\mathrm{W}, n, \ell}\in \cz^{n\times n}$, denoting the finite difference operator and the discrete Haar wavelet \cite{mallat2008wavelet} at level $\ell$, respectively.
For the finite difference operator $\Psi_{\mathrm{TV}, n}$, we use the two-step scheme in Section \ref{sec:TV}. In Step 2), we either use the uniform density according to Corollary \ref{cor:circulant1} (see Fig. \ref{fig:1d_a}), or the distribution \eqref{eq:variable_density} restricted to $\{2,3,\dots, n\}$ computed from the other transform $\Psi_{\mathrm{W}, n, \ell}$ (see Fig. \ref{fig:1d_b}).
For the discrete Haar wavelet $\Psi_{\mathrm{W}, n, \ell}$, we test two distributions on $[n]$ -- the uniform distribution (see Fig. \ref{fig:1d_a}) and the variable density distribution \eqref{eq:variable_density} (see Fig. \ref{fig:1d_b}).

In the numerical experiments, we choose $n = 512$, and synthesize signals $f$ at sparsity levels $s = 16, 32, 48, \dots, 128$ with respect to the above two sparsifying transforms. We use the Haar wavelet at level $\ell = 6$.
To make the experiments more realistic, we synthesize the sparse signals by thresholding a real-world signal (see Fig. \ref{fig:1d_sig}) in the transform domains.
For the discrete Haar wavelet, we analyze the signal using $\Psi_{\mathrm{W}, n, \ell}$, zero out small coefficients to achieve a certain sparsity level, and synthesize the signal using $\Psi_{\mathrm{W}, n, \ell}^*$.
Due to the non-injectivity of the finite difference operator, we replace the last row of $\Psi_{\mathrm{TV}, n}$ by $[1,1,\dots, 1]^\top \in\cz^n$ in the signal analysis and synthesis steps.

For every sampling scheme, we repeat the experiment for $m = 16, 32, 48, \dots, 512$. We run 50 Monte-Carlo experiments for each setting. An instance is declared a success if the following reconstruction signal-to-noise ratio (RSNR) of the solution $\hat{g}$ to \eqref{eq:ell1min} exceeds $60$dB.
\[
\mathrm{RSNR} = -20\log_{10}\left(\frac{\norm{\hat{g}-x}_2}{\norm{x}_2}\right).
\]
The success rate is computed for every pair $(s, m)$ and shown in Fig. \ref{fig:1d_c} -- \ref{fig:1d_f}.

For signal recovery using the Haar wavelet, the variable density sampling scheme \eqref{eq:variable_density} yields higher success rate than the uniform density sampling scheme (see Fig. \ref{fig:1d_c} and \ref{fig:1d_d}). For 1D TV, the density \eqref{eq:variable_density} computed using 1D TV, which is uniform, yields higher success rate than that computed from the Haar wavelet (see Fig. \ref{fig:1d_e} and \ref{fig:1d_f}). These experiments show that one can recover sparse signals more successfully using the sampling density \eqref{eq:variable_density} computed from the specific sparsifying transform used in the recovery, which is adaptive to local incoherence between the transform and the measurement.
It is commonly believed among practitioners that sampling more densely in the low frequency region, where the energy of a natural signal is concentrated, always provides a better reconstruction. However, that turns out not be the case when one cares about perfect recovery of exactly sparse signals in a certain transform domain from noise-free measurements.

\begin{figure}[ht]
\centering
\subfloat[\label{fig:1d_sig}]{
\input{1d_sig.tex}
}\quad
\subfloat[\label{fig:2d_sig}]{
\includegraphics[width=0.25\columnwidth]{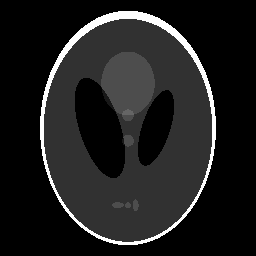}
}
\caption{Signals used in the experiments. (a) A line in an MRI brain image, used to synthesize 1D sparse signals. We only plot the absolute value of this complex signal. (b) The modified Shepp-Logan phantom ($256\times 256$), used for 2D signal recovery.}
\label{fig:sig}
\end{figure}

\begin{figure}[ht]
  \centering
  \subfloat[\label{fig:1d_a}]{
%
%
\begin{tikzpicture}[scale=0.8]

\begin{axis}[%
width=2in,
height=1in,
at={(0in,0in)},
scale only axis,
xmin=0,
xmax=512,
xtick={\empty},
ymin=0,
ymax=0.03,
ytick={\empty},
axis background/.style={fill=white}
]
\addplot [color=blue,solid,line width=2.0pt,forget plot]
  table[row sep=crcr]{%
1	0.001953125\\
2	0.001953125\\
3	0.001953125\\
4	0.001953125\\
5	0.001953125\\
6	0.001953125\\
7	0.001953125\\
8	0.001953125\\
9	0.001953125\\
10	0.001953125\\
11	0.001953125\\
12	0.001953125\\
13	0.001953125\\
14	0.001953125\\
15	0.001953125\\
16	0.001953125\\
17	0.001953125\\
18	0.001953125\\
19	0.001953125\\
20	0.001953125\\
21	0.001953125\\
22	0.001953125\\
23	0.001953125\\
24	0.001953125\\
25	0.001953125\\
26	0.001953125\\
27	0.001953125\\
28	0.001953125\\
29	0.001953125\\
30	0.001953125\\
31	0.001953125\\
32	0.001953125\\
33	0.001953125\\
34	0.001953125\\
35	0.001953125\\
36	0.001953125\\
37	0.001953125\\
38	0.001953125\\
39	0.001953125\\
40	0.001953125\\
41	0.001953125\\
42	0.001953125\\
43	0.001953125\\
44	0.001953125\\
45	0.001953125\\
46	0.001953125\\
47	0.001953125\\
48	0.001953125\\
49	0.001953125\\
50	0.001953125\\
51	0.001953125\\
52	0.001953125\\
53	0.001953125\\
54	0.001953125\\
55	0.001953125\\
56	0.001953125\\
57	0.001953125\\
58	0.001953125\\
59	0.001953125\\
60	0.001953125\\
61	0.001953125\\
62	0.001953125\\
63	0.001953125\\
64	0.001953125\\
65	0.001953125\\
66	0.001953125\\
67	0.001953125\\
68	0.001953125\\
69	0.001953125\\
70	0.001953125\\
71	0.001953125\\
72	0.001953125\\
73	0.001953125\\
74	0.001953125\\
75	0.001953125\\
76	0.001953125\\
77	0.001953125\\
78	0.001953125\\
79	0.001953125\\
80	0.001953125\\
81	0.001953125\\
82	0.001953125\\
83	0.001953125\\
84	0.001953125\\
85	0.001953125\\
86	0.001953125\\
87	0.001953125\\
88	0.001953125\\
89	0.001953125\\
90	0.001953125\\
91	0.001953125\\
92	0.001953125\\
93	0.001953125\\
94	0.001953125\\
95	0.001953125\\
96	0.001953125\\
97	0.001953125\\
98	0.001953125\\
99	0.001953125\\
100	0.001953125\\
101	0.001953125\\
102	0.001953125\\
103	0.001953125\\
104	0.001953125\\
105	0.001953125\\
106	0.001953125\\
107	0.001953125\\
108	0.001953125\\
109	0.001953125\\
110	0.001953125\\
111	0.001953125\\
112	0.001953125\\
113	0.001953125\\
114	0.001953125\\
115	0.001953125\\
116	0.001953125\\
117	0.001953125\\
118	0.001953125\\
119	0.001953125\\
120	0.001953125\\
121	0.001953125\\
122	0.001953125\\
123	0.001953125\\
124	0.001953125\\
125	0.001953125\\
126	0.001953125\\
127	0.001953125\\
128	0.001953125\\
129	0.001953125\\
130	0.001953125\\
131	0.001953125\\
132	0.001953125\\
133	0.001953125\\
134	0.001953125\\
135	0.001953125\\
136	0.001953125\\
137	0.001953125\\
138	0.001953125\\
139	0.001953125\\
140	0.001953125\\
141	0.001953125\\
142	0.001953125\\
143	0.001953125\\
144	0.001953125\\
145	0.001953125\\
146	0.001953125\\
147	0.001953125\\
148	0.001953125\\
149	0.001953125\\
150	0.001953125\\
151	0.001953125\\
152	0.001953125\\
153	0.001953125\\
154	0.001953125\\
155	0.001953125\\
156	0.001953125\\
157	0.001953125\\
158	0.001953125\\
159	0.001953125\\
160	0.001953125\\
161	0.001953125\\
162	0.001953125\\
163	0.001953125\\
164	0.001953125\\
165	0.001953125\\
166	0.001953125\\
167	0.001953125\\
168	0.001953125\\
169	0.001953125\\
170	0.001953125\\
171	0.001953125\\
172	0.001953125\\
173	0.001953125\\
174	0.001953125\\
175	0.001953125\\
176	0.001953125\\
177	0.001953125\\
178	0.001953125\\
179	0.001953125\\
180	0.001953125\\
181	0.001953125\\
182	0.001953125\\
183	0.001953125\\
184	0.001953125\\
185	0.001953125\\
186	0.001953125\\
187	0.001953125\\
188	0.001953125\\
189	0.001953125\\
190	0.001953125\\
191	0.001953125\\
192	0.001953125\\
193	0.001953125\\
194	0.001953125\\
195	0.001953125\\
196	0.001953125\\
197	0.001953125\\
198	0.001953125\\
199	0.001953125\\
200	0.001953125\\
201	0.001953125\\
202	0.001953125\\
203	0.001953125\\
204	0.001953125\\
205	0.001953125\\
206	0.001953125\\
207	0.001953125\\
208	0.001953125\\
209	0.001953125\\
210	0.001953125\\
211	0.001953125\\
212	0.001953125\\
213	0.001953125\\
214	0.001953125\\
215	0.001953125\\
216	0.001953125\\
217	0.001953125\\
218	0.001953125\\
219	0.001953125\\
220	0.001953125\\
221	0.001953125\\
222	0.001953125\\
223	0.001953125\\
224	0.001953125\\
225	0.001953125\\
226	0.001953125\\
227	0.001953125\\
228	0.001953125\\
229	0.001953125\\
230	0.001953125\\
231	0.001953125\\
232	0.001953125\\
233	0.001953125\\
234	0.001953125\\
235	0.001953125\\
236	0.001953125\\
237	0.001953125\\
238	0.001953125\\
239	0.001953125\\
240	0.001953125\\
241	0.001953125\\
242	0.001953125\\
243	0.001953125\\
244	0.001953125\\
245	0.001953125\\
246	0.001953125\\
247	0.001953125\\
248	0.001953125\\
249	0.001953125\\
250	0.001953125\\
251	0.001953125\\
252	0.001953125\\
253	0.001953125\\
254	0.001953125\\
255	0.001953125\\
256	0.001953125\\
257	0.001953125\\
258	0.001953125\\
259	0.001953125\\
260	0.001953125\\
261	0.001953125\\
262	0.001953125\\
263	0.001953125\\
264	0.001953125\\
265	0.001953125\\
266	0.001953125\\
267	0.001953125\\
268	0.001953125\\
269	0.001953125\\
270	0.001953125\\
271	0.001953125\\
272	0.001953125\\
273	0.001953125\\
274	0.001953125\\
275	0.001953125\\
276	0.001953125\\
277	0.001953125\\
278	0.001953125\\
279	0.001953125\\
280	0.001953125\\
281	0.001953125\\
282	0.001953125\\
283	0.001953125\\
284	0.001953125\\
285	0.001953125\\
286	0.001953125\\
287	0.001953125\\
288	0.001953125\\
289	0.001953125\\
290	0.001953125\\
291	0.001953125\\
292	0.001953125\\
293	0.001953125\\
294	0.001953125\\
295	0.001953125\\
296	0.001953125\\
297	0.001953125\\
298	0.001953125\\
299	0.001953125\\
300	0.001953125\\
301	0.001953125\\
302	0.001953125\\
303	0.001953125\\
304	0.001953125\\
305	0.001953125\\
306	0.001953125\\
307	0.001953125\\
308	0.001953125\\
309	0.001953125\\
310	0.001953125\\
311	0.001953125\\
312	0.001953125\\
313	0.001953125\\
314	0.001953125\\
315	0.001953125\\
316	0.001953125\\
317	0.001953125\\
318	0.001953125\\
319	0.001953125\\
320	0.001953125\\
321	0.001953125\\
322	0.001953125\\
323	0.001953125\\
324	0.001953125\\
325	0.001953125\\
326	0.001953125\\
327	0.001953125\\
328	0.001953125\\
329	0.001953125\\
330	0.001953125\\
331	0.001953125\\
332	0.001953125\\
333	0.001953125\\
334	0.001953125\\
335	0.001953125\\
336	0.001953125\\
337	0.001953125\\
338	0.001953125\\
339	0.001953125\\
340	0.001953125\\
341	0.001953125\\
342	0.001953125\\
343	0.001953125\\
344	0.001953125\\
345	0.001953125\\
346	0.001953125\\
347	0.001953125\\
348	0.001953125\\
349	0.001953125\\
350	0.001953125\\
351	0.001953125\\
352	0.001953125\\
353	0.001953125\\
354	0.001953125\\
355	0.001953125\\
356	0.001953125\\
357	0.001953125\\
358	0.001953125\\
359	0.001953125\\
360	0.001953125\\
361	0.001953125\\
362	0.001953125\\
363	0.001953125\\
364	0.001953125\\
365	0.001953125\\
366	0.001953125\\
367	0.001953125\\
368	0.001953125\\
369	0.001953125\\
370	0.001953125\\
371	0.001953125\\
372	0.001953125\\
373	0.001953125\\
374	0.001953125\\
375	0.001953125\\
376	0.001953125\\
377	0.001953125\\
378	0.001953125\\
379	0.001953125\\
380	0.001953125\\
381	0.001953125\\
382	0.001953125\\
383	0.001953125\\
384	0.001953125\\
385	0.001953125\\
386	0.001953125\\
387	0.001953125\\
388	0.001953125\\
389	0.001953125\\
390	0.001953125\\
391	0.001953125\\
392	0.001953125\\
393	0.001953125\\
394	0.001953125\\
395	0.001953125\\
396	0.001953125\\
397	0.001953125\\
398	0.001953125\\
399	0.001953125\\
400	0.001953125\\
401	0.001953125\\
402	0.001953125\\
403	0.001953125\\
404	0.001953125\\
405	0.001953125\\
406	0.001953125\\
407	0.001953125\\
408	0.001953125\\
409	0.001953125\\
410	0.001953125\\
411	0.001953125\\
412	0.001953125\\
413	0.001953125\\
414	0.001953125\\
415	0.001953125\\
416	0.001953125\\
417	0.001953125\\
418	0.001953125\\
419	0.001953125\\
420	0.001953125\\
421	0.001953125\\
422	0.001953125\\
423	0.001953125\\
424	0.001953125\\
425	0.001953125\\
426	0.001953125\\
427	0.001953125\\
428	0.001953125\\
429	0.001953125\\
430	0.001953125\\
431	0.001953125\\
432	0.001953125\\
433	0.001953125\\
434	0.001953125\\
435	0.001953125\\
436	0.001953125\\
437	0.001953125\\
438	0.001953125\\
439	0.001953125\\
440	0.001953125\\
441	0.001953125\\
442	0.001953125\\
443	0.001953125\\
444	0.001953125\\
445	0.001953125\\
446	0.001953125\\
447	0.001953125\\
448	0.001953125\\
449	0.001953125\\
450	0.001953125\\
451	0.001953125\\
452	0.001953125\\
453	0.001953125\\
454	0.001953125\\
455	0.001953125\\
456	0.001953125\\
457	0.001953125\\
458	0.001953125\\
459	0.001953125\\
460	0.001953125\\
461	0.001953125\\
462	0.001953125\\
463	0.001953125\\
464	0.001953125\\
465	0.001953125\\
466	0.001953125\\
467	0.001953125\\
468	0.001953125\\
469	0.001953125\\
470	0.001953125\\
471	0.001953125\\
472	0.001953125\\
473	0.001953125\\
474	0.001953125\\
475	0.001953125\\
476	0.001953125\\
477	0.001953125\\
478	0.001953125\\
479	0.001953125\\
480	0.001953125\\
481	0.001953125\\
482	0.001953125\\
483	0.001953125\\
484	0.001953125\\
485	0.001953125\\
486	0.001953125\\
487	0.001953125\\
488	0.001953125\\
489	0.001953125\\
490	0.001953125\\
491	0.001953125\\
492	0.001953125\\
493	0.001953125\\
494	0.001953125\\
495	0.001953125\\
496	0.001953125\\
497	0.001953125\\
498	0.001953125\\
499	0.001953125\\
500	0.001953125\\
501	0.001953125\\
502	0.001953125\\
503	0.001953125\\
504	0.001953125\\
505	0.001953125\\
506	0.001953125\\
507	0.001953125\\
508	0.001953125\\
509	0.001953125\\
510	0.001953125\\
511	0.001953125\\
512	0.001953125\\
};
\end{axis}
\end{tikzpicture}%
  }\quad
  \subfloat[\label{fig:1d_b}]{
  \input{den_variable.tex}
  }\\
  \subfloat[\label{fig:1d_c}]{
%
%
\begin{tikzpicture}[scale=0.8]

\begin{axis}[%
width=2in,
height=2in,
at={(0in,0in)},
scale only axis,
point meta min=0,
point meta max=1,
axis on top,
xmin=0.5,
xmax=16.5,
xtick={4,8,12,16},
xticklabels={{0.125},{0.25},{0.375},{0.5}},
ymin=0.5,
ymax=32.5,
ytick={8,16,24,32},
yticklabels={{0.25},{0.5},{0.75},{1}},
axis background/.style={fill=white}
]
\addplot [forget plot] graphics [xmin=0.5,xmax=16.5,ymin=0.5,ymax=32.5] {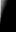};
\end{axis}
\end{tikzpicture}%
  }\quad
  \subfloat[\label{fig:1d_d}]{
%
%
\begin{tikzpicture}[scale=0.8]

\begin{axis}[%
width=2in,
height=2in,
at={(0in,0in)},
scale only axis,
point meta min=0,
point meta max=1,
axis on top,
xmin=0.5,
xmax=16.5,
xtick={4,8,12,16},
xticklabels={{0.125},{0.25},{0.375},{0.5}},
ymin=0.5,
ymax=32.5,
ytick={8,16,24,32},
yticklabels={{0.25},{0.5},{0.75},{1}},
axis background/.style={fill=white}
]
\addplot [forget plot] graphics [xmin=0.5,xmax=16.5,ymin=0.5,ymax=32.5] {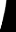};
\end{axis}
\end{tikzpicture}%
  }\\
  \subfloat[\label{fig:1d_e}]{
%
%
\begin{tikzpicture}[scale=0.8]

\begin{axis}[%
width=2in,
height=2in,
at={(0in,0in)},
scale only axis,
point meta min=0,
point meta max=1,
axis on top,
xmin=0.5,
xmax=16.5,
xtick={4,8,12,16},
xticklabels={{0.125},{0.25},{0.375},{0.5}},
ymin=0.5,
ymax=32.5,
ytick={8,16,24,32},
yticklabels={{0.25},{0.5},{0.75},{1}},
axis background/.style={fill=white}
]
\addplot [forget plot] graphics [xmin=0.5,xmax=16.5,ymin=0.5,ymax=32.5] {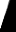};
\end{axis}
\end{tikzpicture}%
  }\quad
  \subfloat[\label{fig:1d_f}]{
%
%
\begin{tikzpicture}[scale=0.8]

\begin{axis}[%
width=2in,
height=2in,
at={(0in,0in)},
scale only axis,
point meta min=0,
point meta max=1,
axis on top,
xmin=0.5,
xmax=16.5,
xtick={4,8,12,16},
xticklabels={{0.125},{0.25},{0.375},{0.5}},
ymin=0.5,
ymax=32.5,
ytick={8,16,24,32},
yticklabels={{0.25},{0.5},{0.75},{1}},
axis background/.style={fill=white}
]
\addplot [forget plot] graphics [xmin=0.5,xmax=16.5,ymin=0.5,ymax=32.5] {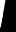};
\end{axis}
\end{tikzpicture}%
  }
  \caption{Sampling densities and empirical success rate for 1D signal recovery. (a) (b) The sampling density \eqref{eq:variable_density} computed from 1D TV (uniform), and from the Haar wavelet. (c) -- (f) The success rate plots. The $x$-axis represents the sparsity level $s/n$, and the $y$-axis represents the sub-sampling ratio $m/n$. The empirical success rate is represented by the grayscale value (white for $1$ and black for $0$).  (c) (d) Recovery using the Haar wavelet, under sampling density (a) and (b), respectively. (e) (f) Recovery using 1D TV, under sampling density (a) and (b), respectively, both restricted to $\{2,3,\dots,n\}$.}%
  \label{fig:1d}%
\end{figure}

\begin{figure}[htbp]
  \centering
  \subfloat[\label{fig:2d_a}]{
  \includegraphics[width=0.25\columnwidth]{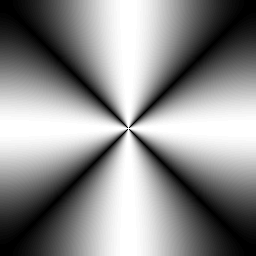}
  }\quad
  \subfloat[\label{fig:2d_b}]{
  \includegraphics[width=0.25\columnwidth]{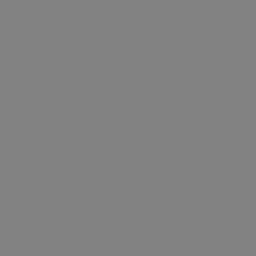}
  }\quad
  \subfloat[\label{fig:2d_c}]{
  \includegraphics[width=0.25\columnwidth]{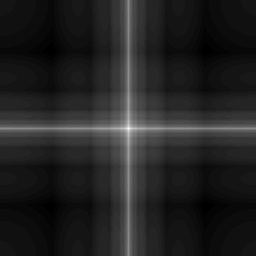}
  }\\
  \subfloat[\label{fig:2d_d}]{
%
%
\begin{tikzpicture}[scale = 0.8]

\begin{axis}[%
width=2in,
height=1in,
at={(0in,0in)},
scale only axis,
xmin=1,
xmax=8,
xtick={1,2,3,4,5,6,7,8},
xticklabels={{1/32},{1/16},{3/32},{1/8},{5/32},{3/16},{7/32},{1/4}},
ymin=0,
ymax=1,
ytick={0,0.2,0.4,0.6,0.8,1},
axis background/.style={fill=white}
]
\addplot [color=red,solid,line width=2.0pt,forget plot]
  table[row sep=crcr]{%
1	0\\
2	1\\
3	1\\
4	1\\
5	1\\
6	1\\
7	1\\
8	1\\
};
\addplot [color=blue,dashed,line width=2.0pt,forget plot]
  table[row sep=crcr]{%
1	0\\
2	0.88\\
3	1\\
4	1\\
5	1\\
6	1\\
7	1\\
8	1\\
};
\end{axis}
\end{tikzpicture}%
  }\quad
  \subfloat[\label{fig:2d_e}]{
%
%
\begin{tikzpicture}[scale = 0.8]

\begin{axis}[%
width=2in,
height=1in,
at={(0in,0in)},
scale only axis,
xmin=1,
xmax=8,
xtick={1,2,3,4,5,6,7,8},
xticklabels={{1/32},{1/16},{3/32},{1/8},{5/32},{3/16},{7/32},{1/4}},
ymin=0,
ymax=1,
ytick={0,0.2,0.4,0.6,0.8,1},
axis background/.style={fill=white}
]
\addplot [color=red,solid,line width=2.0pt,forget plot]
  table[row sep=crcr]{%
1	0\\
2	0.26\\
3	1\\
4	1\\
5	1\\
6	1\\
7	1\\
8	1\\
};
\addplot [color=blue,dashed,line width=2.0pt,forget plot]
  table[row sep=crcr]{%
1	0\\
2	0\\
3	0\\
4	0\\
5	0.38\\
6	1\\
7	1\\
8	1\\
};
\end{axis}
\end{tikzpicture}%
  }
  \caption{Sampling density and empirical success rate for 2D signal recovery. (a) -- (c) The sampling densities on a $256\times 256$ grid, computed from 2D anisotropic TV, 2D isotropic TV (uniform), and 2D separable Haar wavelet. Grayscale values represent samping density in log scale. (d) (e) The empirical success rate for 2D signal recovery using anisotropic TV minimization and isotropic TV minimization, respectively. The $x$-axis represents the sub-sampling ratio $m/n$, and the $y$-axis represents the empirical success rate. Within the same plot, different curves represent the results using different sampling schemes. The solid red curves in (d) and (e) represent results using sampling density (a) and (b), respectively. The dashed blue curves represent results using sampling density (c) computed from wavelet.}%
  \label{fig:2d}%
\end{figure}

\subsection{2D Signals}
We also run numerical experiments on a 2D signal, the modified Shepp-Logan phantom (see Fig. \ref{fig:2d_sig}). We minimize the anisotropic and isotropic 2D total variations, which correspond to solving \eqref{eq:ell1min} and \eqref{eq:mixedmin} for the 2D finite difference operator $\Phi_{\mathrm{TV},n_1,n_2}$.

For the two types of total variations, we use the two-step scheme in Section \ref{sec:TV}. In Step 2), we use the densities in \eqref{eq:variable_density} and \eqref{eq:variable_density_gs} for anisotropic and isotropic total variations, respectively. As a comparison, we use the variable density proposed by Krahmer and Ward \cite{krahmer2014stable}, which is \eqref{eq:variable_density} computed with the 2D separable Haar wavelet $\Psi_{\mathrm{W}, n_2, \log_2n_2}\otimes \Psi_{\mathrm{W}, n_1, \log_2n_1}$.

We use a phantom of size $n = n_1\times n_2 = 256\times 256$, and repeat the experiments for $m= n/32, n/16,3n/32,\dots, n/4$. We run 50 Monte-Carlo experiments for each setting, and compute the success rate for every chioce of $m$, as shown in Fig. \ref{fig:2d_d} and \ref{fig:2d_e}.

For both anisotropic TV minimization \eqref{eq:ell1min} and isotropic TV minimization \eqref{eq:mixedmin}, the success rates using sampling density computed from TV are higher than those using sampling density computed from separable Haar wavelet. Although the performances of two sampling schemes are relatively close for anisotropic TV minimization (see Fig. \ref{fig:2d_d}), the advantage of the density computed from TV over that computed from wavelet is more pronounced for isotropic TV minimization (see Fig. \ref{fig:2d_e}). Contrary to common belief, sampling low frequencies more densely does not always lead to superior recovery. We suggest using sampling densities \eqref{eq:variable_density} and \eqref{eq:variable_density_gs} tailored to the specific sparsifying transform and the measurement operator.

\section{Conclusion}
\label{sec:concl}

In this paper, we established a unified theory for recovery of sparse signals in a transform domain. Our theory guaranteed robust recovery from noisy measurements by convex programming and apply without relying on a particular choice of measurement and sparsifying transforms. We quantified the sufficient sampling rate using functions of the two transforms, and this result identifies a class of measurement and sparsity models enabling recovery at a near optimal sampling rate.
We also proposed a variable sampling density designed with incoherence parameters of the two transforms, which provided recovery guarantee at a lower sampling rate than previous works in various scenarios.
Furthermore, we extended the result to the group-sparsity models so that it also applies to the popular isotropic total variation minimization.
In particular, for the partial Fourier recovery of sparse signals over a circulant transform, our theory suggests a uniformly random sampling or its variation.
Our numerical results showed that our variable density random sampling strategy outperforms other known sampling strategies in various scenarios. This suggests that our new theory is indeed universally useful.

\section*{Acknowledgement}
The authors thank referees for their valuable comments and suggestions.

\appendices

\section{Bernstein inequalities}

\begin{theorem}[{Matrix Bernstein Inequality \cite{tropp2012user}}]
\label{thm:mtx_bernstein_ineq}
Let $\{X_j\} \in \cz^{d \times d}$ be a finite sequence of independent random matrices.
Suppose that $\mathbb{E} X_j = 0$ and $\norm{X_j} \leq B$ almost surely for all $j$
and
\[
\max\left( \left\| \sum_j \mathbb{E} X_j X_j^* \right\|,~ \left\| \sum_j \mathbb{E} X_j^* X_j \right\| \right) \leq \sigma^2.
\]
Then for all $t \geq 0$,
\[
\mathbb{P} \left( \left\| \sum_j X_j \right\| \geq t \right)
\leq 2d \exp \left( \frac{-t^2/2}{\sigma^2 + Bt/3} \right).
\]
\end{theorem}

\begin{theorem}[{Vector Bernstein Inequality \cite{candes2011probabilistic}}]
\label{thm:vec_bernstein_ineq}
Let $\{v_j\} \in \cz^d$ be a finite sequence of independent random vectors.
Suppose that $\mathbb{E} v_j = 0$ and $\norm{v_j}_2 \leq B$ almost surely for all $j$
and $\mathbb{E} \sum_j \norm{v_j}_2^2 \leq \sigma^2$.
Then for all $0 \leq t \leq \sigma^2/B$,
\[
\mathbb{P} \left( \left\| \sum_j v_j \right\|_2 \geq t \right)
\leq \exp \left( - \frac{t^2}{8 \sigma^2} + \frac{1}{4} \right).
\]
\end{theorem}

\section{Proof of Lemma~\ref{lemma:E1}}
\label{sec:pf:lemma:E1}
Define
\[
X_j := \Pi_J (n T e_{\omega_j} e_{\omega_j}^* T^\dagger - T T^\dagger) \Pi_J, \quad \forall j \in [m].
\]
Then, $X_j$ satisfies $\mathbb{E} X_j = 0$ and $\norm{X_j} \leq \mu s$ for all $j$.
Since
\begin{align*}
X_j^* X_j
{} & = n^2 \Pi_J \widetilde{T} e_{\omega_j} e_{\omega_j}^* T^* \Pi_J T e_{\omega_j} e_{\omega_j}^* T^\dagger \Pi_J \\
{} & \quad - n \Pi_J \widetilde{T} e_{\omega_j} e_{\omega_j}^* T^* \Pi_J T T^\dagger \Pi_J \\
{} & \quad - n \Pi_J T T^\dagger \Pi_J T e_{\omega_j} e_{\omega_j}^* T^\dagger \Pi_J \\
{} & \quad + \Pi_J T T^\dagger \Pi_J T T^\dagger \Pi_J,
\end{align*}
it follows that
\begin{align*}
\mathbb{E} X_j^* X_j
{} & = \mathbb{E} n^2 \Pi_J \widetilde{T} e_{\omega_j} e_{\omega_j}^* T^* \Pi_J T e_{\omega_j} e_{\omega_j}^* T^\dagger \Pi_J - \Pi_J T T^\dagger \Pi_J T T^\dagger \Pi_J \\
{} & \leq \mathbb{E} n^2 e_{\omega_j}^* T^* \Pi_J T e_{\omega_j} \Pi_J \widetilde{T} e_{\omega_j} e_{\omega_j}^* T^\dagger \Pi_J \\
{} & = \mathbb{E} n^2 \norm{\gamma^{1/2} \Pi_J T e_{\omega_j}}_2^2 \gamma^{-1} \Pi_J \widetilde{T} e_{\omega_j} e_{\omega_j}^* T^\dagger \Pi_J \\
{} & \leq \mathbb{E} n \mu s \gamma^{-1} \Pi_J \widetilde{T} e_{\omega_j} e_{\omega_j}^* T^\dagger \Pi_J \\
{} & = \mu s \gamma^{-1} \Pi_J \widetilde{T} T^\dagger \Pi_J.
\end{align*}
By symmetry, we also have
\[
\mathbb{E} X_j X_j^* \leq \mu s \gamma \Pi_J T T^* \Pi_J.
\]
Therefore,
\[
\max\left( \left\| \sum_{j=1}^m \mathbb{E} X_j X_j^* \right\|,~ \left\| \sum_{j=1}^m \mathbb{E} X_j^* X_j \right\| \right)
\leq m \mu s \max( \norm{\gamma^{-1} \widetilde{T} T^\dagger}, \norm{\gamma T T^*} )
\leq \frac{m \mu s}{1 - \norm{\gamma T^* T - I_n}}.
\]
Applying the above results to Theorem~\ref{thm:mtx_bernstein_ineq} with $t = m \delta$ completes the proof.

\section{Proof of Lemma~\ref{lemma:E2}}
\label{sec:pf:lemma:E2}
Define
\[
v_j := \Pi_J (n T e_{\omega_j} e_{\omega_j}^* T^\dagger - T T^\dagger) \Pi_J q, \quad \forall j \in [m].
\]
Then, $v_j$ satisfies $\mathbb{E} v_j = 0$ and
\begin{align*}
\norm{v_j}_2
{} & \leq \norm{\Pi_J n T e_{\omega_j} e_{\omega_j}^* T^\dagger \Pi_J q}_2 + \norm{\Pi_J T T^\dagger \Pi_J q}_2 \\
{} & \leq (\norm{\gamma^{1/2} \Pi_J \sqrt{n} T e_{\omega_j}}_2 \norm{\gamma^{-1/2} \Pi_J \sqrt{n} \widetilde{T} e_{\omega_j}}_2 + 1) \norm{\Pi_J q}_2 \\
{} & \leq (s\mu+1) \norm{\Pi_J q}_2.
\end{align*}
Furthermore,
\begin{align*}
\mathbb{E} \norm{v_j}_2^2
{} & = \mathbb{E} n^2 q^* \Pi_J \widetilde{T} e_{\omega_j} e_{\omega_j}^* T^* \Pi_J T e_{\omega_j} e_{\omega_j}^* T^\dagger \Pi_J q
- q^* \Pi_J T T^\dagger \Pi_J q \\
{} & \leq n \norm{\gamma^{1/2} \Pi_J T e_{\omega_j}}_2^2 \mathbb{E} n \gamma^{-1} q^* \Pi_J \widetilde{T} e_{\omega_j} e_{\omega_j}^* T^\dagger \Pi_J q \\
{} & \leq \mu s \gamma^{-1} \mathbb{E} n q^* \Pi_J \widetilde{T} e_{\omega_j} e_{\omega_j}^* T^\dagger \Pi_J q \\
{} & = \mu s \gamma^{-1} q^* \Pi_J \widetilde{T} T^\dagger \Pi_J q \\
{} & \leq \frac{\mu s \norm{\Pi_J q}_2^2}{1 - \norm{\gamma T^* T - I_n}},
\end{align*}
where the second inequality follows by the incoherence property.
Applying the above results to Theorem~\ref{thm:vec_bernstein_ineq} completes the proof.

\section{Proof of Lemma~\ref{lemma:E3}}
\label{sec:pf:lemma:E3}
Let $i \in [n] \setminus J$ be arbitrarily fixed.
Define
\[
w_j := \langle e_i, (n T e_{\omega_j} e_{\omega_j}^* T^\dagger - T T^\dagger) \Pi_J q \rangle, \quad \forall j \in [m].
\]
Then, $w_j$ satisfies $\mathbb{E} w_j = 0$ and
\begin{align*}
|w_j|
{} & \leq |e_j^* n T e_{\omega_j} e_{\omega_j}^* T^\dagger \Pi_J q| + \norm{T T^\dagger \Pi_J q}_\infty \\
{} & \leq (\norm{\gamma^{1/2} \sqrt{n} T e_{\omega_j}}_\infty \norm{\gamma^{-1/2} \sqrt{n} \Pi_J \widetilde{T} e_{\omega_j}}_2 + 1) \norm{\Pi_J q}_2 \\
{} & \leq (\sqrt{s} \mu+1) \norm{\Pi_J q}_2.
\end{align*}
Furthermore,
\begin{align*}
\mathbb{E} |w_j|^2
{} & = \mathbb{E} n^2 q^* \Pi_J \widetilde{T} e_{\omega_j} e_{\omega_j}^* T^* e_i e_i^* T e_{\omega_j} e_{\omega_j}^* T^\dagger \Pi_J q
- q^* \Pi_J T T^\dagger e_i e_i^* T T^\dagger \Pi_J q \\
{} & \leq n \norm{\gamma^{1/2} T e_{\omega_j}}_\infty^2 \gamma^{-1} \mathbb{E} n q^* \Pi_J \widetilde{T} e_{\omega_j} e_{\omega_j}^* T^\dagger \Pi_J q \\
{} & \leq \mu \gamma^{-1} q^* \Pi_J \widetilde{T} T^\dagger \Pi_J q \\
{} & \leq \frac{\mu \norm{\Pi_J q}_2^2}{1 - \norm{\gamma T^* T - I_n}}.
\end{align*}
Applying the above results to Theorem~\ref{thm:mtx_bernstein_ineq} gives
\[
\mathbb{P}\left(\left| \left\langle e_i, \left(\frac{n}{m} B_\Omega^* - B\right) \Pi_J q \right\rangle \right| \geq t \norm{\Pi_J q}_2\right)
\leq \exp\left( - \frac{m}{2\mu} \cdot \frac{t^2}{1/(1-\norm{\gamma T^* T - I_n}) + \sqrt{s}t/3} \right).
\]
Combine this for $i \in [N]$ with the union bound completes the proof.

\section{Proof of Lemma~\ref{lemma:E3'}}
\label{sec:pf:lemma:E3'}
Let $i \in [n] \setminus J$ be arbitrarily fixed.
Define
\[
v_j := \Pi_{\calG_i} (n T e_{\omega_j} e_{\omega_j}^* T^\dagger - T T^\dagger) \Pi_{\calG_J} q, \quad \forall j \in [m].
\]
Then, $v_j$ satisfies $\mathbb{E} v_j = 0$ and
\begin{align*}
\norm{v_j}_2
{} & \leq \norm{\Pi_{\calG_i} n T e_{\omega_j} e_{\omega_j}^* T^\dagger \Pi_{\calG_J} q}_2 + \norm{\Pi_{\calG_i} T T^\dagger \Pi_{\calG_J} q}_2 \\
{} & \leq (\norm{\gamma^{1/2} \sqrt{n} \Pi_{\calG_i} T e_{\omega_j}}_2 \norm{\gamma^{-1/2} \sqrt{n} \Pi_{\calG_J} \widetilde{T} e_{\omega_j}}_2 + 1) \norm{\Pi_{\calG_J} q}_2 \\
{} & \leq (\sqrt{s} \mu_\calG + 1) \norm{\Pi_{\calG_J} q}_2.
\end{align*}
Furthermore,
\begin{align*}
\mathbb{E} v_j^* v_j
{} & = \mathbb{E} n^2 q^* \Pi_{\calG_J} \widetilde{T} e_{\omega_j} e_{\omega_j}^* T^* \Pi_{\calG_i} T e_{\omega_j} e_{\omega_j}^* T^\dagger \Pi_{\calG_J} q
- q^* \Pi_{\calG_J} T T^\dagger \Pi_{\calG_i} T T^\dagger \Pi_{\calG_J} q \\
{} & \leq n \norm{\gamma^{1/2} \Pi_{\calG_i} T e_{\omega_j}}_2^2 \gamma^{-1} \mathbb{E} n q^* \Pi_{\calG_J} \widetilde{T} e_{\omega_j} e_{\omega_j}^* T^\dagger \Pi_{\calG_J} q \\
{} & \leq \mu_\calG \gamma^{-1} q^* \Pi_{\calG_J} \widetilde{T} T^\dagger \Pi_{\calG_J} q \\
{} & \leq \frac{\mu_\calG \norm{\Pi_{\calG_J} q}_2^2}{1 - \norm{\gamma T^* T - I_n}}.
\end{align*}
Note that
\[
v_j v_j^* \leq v_j^* v_j I_L, \quad \forall j \in [m].
\]
Applying the above results to Theorem~\ref{thm:vec_bernstein_ineq} gives
\[
\mathbb{P}\left(\left\| \Pi_{\calG_i} \left(\frac{n}{m} B_\Omega^* - B\right) \Pi_{\calG_J} q \right\|_2 \geq t \norm{\Pi_{\calG_J} q}_2\right)
\leq \exp\left( - \frac{m}{2\mu} \cdot \frac{t^2}{1/(1-\norm{\gamma T^* T - I_n}) + \sqrt{s}t/3} \right).
\]
Combine this for $i \in [N]$ with the union bound completes the proof.

\section{Proof of Theorem~\ref{thm:rboplike}}
\label{sec:pf:thm:rboplike}
Theorem~\ref{thm:rboplike} is analogous to \cite[Theorem~3.1]{lee2013oblique}.
It has been shown that if $T$ is of full row rank,
then the deviation of $\frac{n}{m} T S_\Omega^* S_\Omega T^\dagger$ from $T T^\dagger = I_N$
is small with high probability \cite[Theorem~3.1]{lee2013oblique}.
On the contrary, Theorem~\ref{thm:rboplike} assumes that $T$ is of full column rank
and shows that the deviation of $\frac{n}{m} T S_\Omega^* S_\Omega T^\dagger$ from $T T^\dagger$, which is not necessarily $I_N$,
is small with high probability.

The proof of Theorem~\ref{thm:rboplike} is obtained from that of \cite[Theorem~3.1]{lee2013oblique} by replacing $I_N$ by $T T^\dagger$.
For example, the isotropy condition
\[
\frac{n}{m} \mathbb{E} T S_\Omega^* S_\Omega T^\dagger = I_N
\]
is replaced by
\[
\frac{n}{m} \mathbb{E} T S_\Omega^* S_\Omega T^\dagger = T T^\dagger.
\]
For a matrix $M \in \cz^{N \times N}$, the term $\theta_s(M)$, previously defined by in \cite{lee2013oblique}
\[
\theta_s(M) := \max_{|\widetilde{J}| \leq s} \norm{\Pi_{\widetilde{J}} (M - I_N) \Pi_{\widetilde{J}}}
\]
is replaced by
\[
\theta_s(M) := \max_{|\widetilde{J}| \leq s} \norm{\Pi_{\widetilde{J}} (M - T T^\dagger) \Pi_{\widetilde{J}}}
\]
Clearly, $T T^\dagger = I_N$ if $T$ has full row rank.
However, in the hypothesis of Theorem~\ref{thm:rboplike}, $T$ has full column rank and $T T^\dagger$ may be rank deficient.
Since the modifications are rather straightforward, we omit the details of the proof and refer them to \cite[Appendix~E]{lee2013oblique}.

By modifying \cite[Theorem~3.1]{lee2013oblique} and its proof as shown above,
\eqref{eq:rboplike} is implied by \eqref{eq:rboplike:cond2} and
\[
m \geq C_1 \delta^{-2} K_T \mu s \log^2 s \log N \log m,
\]
where the factor $K_T$ is given by
\begin{align*}
K_T
{} & = \left\{
\left(2 + \max_{|\widetilde{J}| \leq s} \left\|\Pi_{\widetilde{J}} (\gamma T T^* - T T^\dagger) \Pi_{\widetilde{J}}\right\|\right)^{1/2}
+
\left(2 + \max_{|\widetilde{J}| \leq s} \left\|\Pi_{\widetilde{J}} (\gamma^{-1} \widetilde{T} T^\dagger - T T^\dagger) \Pi_{\widetilde{J}}\right\|\right)^{1/2}
\right\}^2 \\
{} & \leq 4 + 2 \max\left(
\max_{|\widetilde{J}| \leq s} \norm{\Pi_{\widetilde{J}}(\gamma T T^* - T T^\dagger)\Pi_{\widetilde{J}}}
,~
\max_{|\widetilde{J}| \leq s} \norm{\Pi_{\widetilde{J}}(\gamma^{-1} \widetilde{T} T^\dagger - T T^\dagger)\Pi_{\widetilde{J}}}
\right) \\
{} & \leq 4 + 2 \max\left(
\norm{\gamma T T^* - T T^\dagger}
,~
\norm{\gamma^{-1} \widetilde{T} T^\dagger - T T^\dagger}
\right) \\
{} & = 4 + 2 \max\left(
\norm{\gamma T^* T - I_n}
,~
\norm{\gamma^{-1} T^\dagger \widetilde{T} - I_n}
\right).
\end{align*}
Finally, we verify that
\[
\norm{\gamma^{-1} T^\dagger \widetilde{T} - I_n} \leq \frac{1}{\norm{\gamma T^* T - I_n}},
\]
where the upper bound dominates $\norm{\gamma T^* T - I_n}$.
This completes the proof.


\begin{thebibliography}{10}
\providecommand{\url}[1]{#1}
\csname url@samestyle\endcsname
\providecommand{\newblock}{\relax}
\providecommand{\bibinfo}[2]{#2}
\providecommand{\BIBentrySTDinterwordspacing}{\spaceskip=0pt\relax}
\providecommand{\BIBentryALTinterwordstretchfactor}{4}
\providecommand{\BIBentryALTinterwordspacing}{\spaceskip=\fontdimen2\font plus
\BIBentryALTinterwordstretchfactor\fontdimen3\font minus
  \fontdimen4\font\relax}
\providecommand{\BIBforeignlanguage}[2]{{%
\expandafter\ifx\csname l@#1\endcsname\relax
\typeout{** WARNING: IEEEtran.bst: No hyphenation pattern has been}%
\typeout{** loaded for the language `#1'. Using the pattern for}%
\typeout{** the default language instead.}%
\else
\language=\csname l@#1\endcsname
\fi
#2}}
\providecommand{\BIBdecl}{\relax}
\BIBdecl

\bibitem{donoho2006compressed}
D.~L. Donoho, ``Compressed sensing,'' \emph{{IEEE} Trans. Inf. Theory},
  vol.~52, no.~4, pp. 1289--1306, 2006.

\bibitem{candes2006robust}
E.~J. Cand{\`e}s, J.~Romberg, and T.~Tao, ``Robust uncertainty principles:
  Exact signal reconstruction from highly incomplete frequency information,''
  \emph{{IEEE} Trans. Inf. Theory}, vol.~52, no.~2, pp. 489--509, 2006.

\bibitem{beck2009fast}
A.~Beck and M.~Teboulle, ``A fast iterative shrinkage-thresholding algorithm
  for linear inverse problems,'' \emph{SIAM J. Imaging Sci.}, vol.~2, no.~1,
  pp. 183--202, 2009.

\bibitem{boyd2011distributed}
S.~Boyd, N.~Parikh, E.~Chu, B.~Peleato, and J.~Eckstein, ``Distributed
  optimization and statistical learning via the alternating direction method of
  multipliers,'' \emph{Foundations and Trends in Machine Learning}, vol.~3,
  no.~1, pp. 1--122, 2011.

\bibitem{needell2009cosamp}
D.~Needell and J.~A. Tropp, ``{CoSaMP}: Iterative signal recovery from
  incomplete and inaccurate samples,'' \emph{Appl. Comput. Harmon. Anal.},
  vol.~26, no.~3, pp. 301--321, 2009.

\bibitem{dai2009subspace}
W.~Dai and O.~Milenkovic, ``Subspace pursuit for compressive sensing signal
  reconstruction,'' \emph{{IEEE} Trans. Inf. Theory}, vol.~55, no.~5, pp.
  2230--2249, May 2009.

\bibitem{blumensath2009iterative}
T.~Blumensath and M.~E. Davies, ``Iterative hard thresholding for compressed
  sensing,'' \emph{Appl. Comput. Harmon. Anal.}, vol.~27, no.~3, pp. 265--274,
  2009.

\bibitem{foucart2011hard}
S.~Foucart, ``Hard thresholding pursuit: An algorithm for compressive
  sensing,'' \emph{SIAM J. Numer. Anal.}, vol.~49, no.~6, pp. 2543--2563, 2011.

\bibitem{candes2005decoding}
E.~J. Cand{\`e}s and T.~Tao, ``Decoding by linear programming,'' \emph{{IEEE}
  Trans. Inf. Theory}, vol.~51, no.~12, pp. 4203--4215, 2005.

\bibitem{chandrasekaran2012convex}
V.~Chandrasekaran, B.~Recht, P.~A. Parrilo, and A.~S. Willsky, ``The convex
  geometry of linear inverse problems,'' \emph{Found. Comput. Math.}, vol.~12,
  no.~6, pp. 805--849, 2012.

\bibitem{amelunxen2014living}
D.~Amelunxen, M.~Lotz, M.~B. McCoy, and J.~A. Tropp, ``Living on the edge:
  Phase transitions in convex programs with random data,'' \emph{Information
  and Inference}, vol.~3, no.~3, p. 224–294, 2014.

\bibitem{bresler1999image}
Y.~Bresler, M.~Gastpar, and R.~Venkataramani, ``Image compression on-the-fly by
  universal sampling in {F}ourier imaging systems,'' in \emph{Proc. 1999 IEEE
  Information Theory Workshop on Detection, Estimation, Classification, and
  Imaging}, Santa Fe, NM, Feb. 1999, p.~48.

\bibitem{candes2006near}
E.~J. Candes and T.~Tao, ``Near-optimal signal recovery from random
  projections: Universal encoding strategies?'' \emph{{IEEE} Trans. Inf.
  Theory}, vol.~52, no.~12, pp. 5406--5425, 2006.

\bibitem{rudelson2008sparse}
M.~Rudelson and R.~Vershynin, ``On sparse reconstruction from {F}ourier and
  {G}aussian measurements,'' \emph{Comm. Pure Appl. Math.}, vol.~61, no.~8, pp.
  1025--1045, 2008.

\bibitem{candes2011probabilistic}
E.~J. Candes and Y.~Plan, ``A probabilistic and {RIP}less theory of compressed
  sensing,'' \emph{{IEEE} Trans. Inf. Theory}, vol.~57, no.~11, pp. 7235--7254,
  2011.

\bibitem{rauhut2010compressive}
H.~Rauhut, ``Compressive sensing and structured random matrices,'' in
  \emph{Theoretical Foundations and Numerical Methods for Sparse Recovery},
  ser. Radon Series Comp. Appl. Math., M.~Fornasier, Ed.\hskip 1em plus 0.5em
  minus 0.4em\relax Berlin: de Gruyter, 2010, vol.~9, pp. 1--92.

\bibitem{candes2011compressed}
E.~J. Candes, Y.~C. Eldar, D.~Needell, and P.~Randall, ``Compressed sensing
  with coherent and redundant dictionaries,'' \emph{Appl. Comput. Harmon.
  Anal.}, vol.~31, no.~1, pp. 59--73, 2011.

\bibitem{nam2013cosparse}
S.~Nam, M.~E. Davies, M.~Elad, and R.~Gribonval, ``The cosparse analysis model
  and algorithms,'' \emph{Appl. Comput. Harmon. Anal.}, vol.~34, no.~1, pp.
  30--56, 2013.

\bibitem{giryes2014greedy}
R.~Giryes, S.~Nam, M.~Elad, R.~Gribonval, and M.~E. Davies, ``Greedy-like
  algorithms for the cosparse analysis model,'' \emph{Linear Algebra and its
  Applications}, vol. 441, pp. 22--60, 2014.

\bibitem{mallat2008wavelet}
S.~Mallat, \emph{A wavelet tour of signal processing: the sparse way}.\hskip
  1em plus 0.5em minus 0.4em\relax Academic press, 2008.

\bibitem{do2005contourlet}
M.~N. Do and M.~Vetterli, ``The contourlet transform: an efficient directional
  multiresolution image representation,'' \emph{{IEEE} Trans. Image Process.},
  vol.~14, no.~12, pp. 2091--2106, 2005.

\bibitem{starck2002curvelet}
J.-L. Starck, E.~J. Cand{\`e}s, and D.~L. Donoho, ``The curvelet transform for
  image denoising,'' \emph{{IEEE} Trans. Image Process.}, vol.~11, no.~6, pp.
  670--684, 2002.

\bibitem{lustig2008compressed}
M.~Lustig, D.~L. Donoho, J.~M. Santos, and J.~M. Pauly, ``Compressed sensing
  {MRI},'' \emph{{IEEE} Signal Process. Mag.}, vol.~25, no.~2, pp. 72--82,
  2008.

\bibitem{krahmer2011new}
F.~Krahmer and R.~Ward, ``New and improved {Johnson-Lindenstrauss} embeddings
  via the restricted isometry property,'' \emph{SIAM J. Math. Anal.}, vol.~43,
  no.~3, pp. 1269--1281, 2011.

\bibitem{needell2013stable}
D.~Needell and R.~Ward, ``Stable image reconstruction using total variation
  minimization,'' \emph{SIAM J. Imaging Sci.}, vol.~6, no.~2, pp. 1035--1058,
  2013.

\bibitem{krahmer2014stable}
F.~Krahmer and R.~Ward, ``Stable and robust sampling strategies for compressive
  imaging,'' \emph{{IEEE} Trans. Image Process.}, vol.~23, no.~2, pp. 612--622,
  2014.

\bibitem{kabanava2015robust}
M.~Kabanava, H.~Rauhut, and H.~Zhang, ``Robust analysis $\ell_1$-recovery from
  {G}aussian measurements and total variation minimization,'' \emph{European
  Journal of Applied Mathematics}, vol.~26, no.~06, pp. 917--929, 2015.

\bibitem{cai2015guarantees}
J.-F. Cai and W.~Xu, ``Guarantees of total variation minimization for signal
  recovery,'' \emph{Information and Inference}, vol.~4, no.~4, pp. 328--353,
  2015.

\bibitem{kabanava2015analysis}
M.~Kabanava and H.~Rauhut, ``Analysis $\ell_1$-recovery with frames and
  {G}aussian measurements,'' \emph{Acta Applicandae Mathematicae}, vol. 140,
  no.~1, pp. 173--195, 2015.

\bibitem{chen2014robust}
Y.~Chen and Y.~Chi, ``Robust spectral compressed sensing via structured matrix
  completion,'' \emph{{IEEE} Trans. Inf. Theory}, vol.~60, no.~10, pp.
  6576--6601, 2014.

\bibitem{gross2011recovering}
D.~Gross, ``Recovering low-rank matrices from few coefficients in any basis,''
  \emph{{IEEE} Trans. Inf. Theory}, vol.~57, no.~3, pp. 1548--1566, 2011.

\bibitem{lee2013oblique}
K.~Lee, Y.~Bresler, and M.~Junge, ``Oblique pursuits for compressed sensing,''
  \emph{{IEEE} Trans. Inf. Theory}, vol.~59, no.~9, pp. 6111--6141, 2013.

\bibitem{huang2010benefit}
J.~Huang and T.~Zhang, ``The benefit of group sparsity,'' \emph{Ann. Stat.},
  vol.~38, no.~4, pp. 1978--2004, 2010.

\bibitem{pfister2015learning}
L.~Pfister and Y.~Bresler, ``Learning sparsifying filter banks,'' in \emph{SPIE
  Optical Engineering+ Applications}.\hskip 1em plus 0.5em minus 0.4em\relax
  International Society for Optics and Photonics, 2015, pp. 959\,703--959\,703.

\bibitem{ye2016compressive}
J.~C. Ye, J.~M. Kim, K.~H. Jin, and K.~Lee, ``Compressive sampling using
  annihilating filter-based low-rank interpolation,'' \emph{IEEE Transactions
  on Information Theory}, 2016 (in press).

\bibitem{rudelson2013reconstruction}
M.~Rudelson and S.~Zhou, ``Reconstruction from anisotropic random
  measurements,'' \emph{{IEEE} Trans. Inf. Theory}, vol.~59, no.~6, pp.
  3434--3447, 2013.

\bibitem{kueng2014ripless}
R.~Kueng and D.~Gross, ``{RIP}less compressed sensing from anisotropic
  measurements,'' \emph{Linear Algebra Appl.}, vol. 441, pp. 110--123, 2014.

\bibitem{tropp2012user}
J.~A. Tropp, ``User-friendly tail bounds for sums of random matrices,''
  \emph{Found. Comput. Math.}, vol.~12, no.~4, pp. 389--434, 2012.

\end{thebibliography}

\end{document}